\documentclass{article}
\usepackage[utf8]{inputenc}

\usepackage{geometry}
\usepackage{amsthm}
\usepackage{amsmath}
\usepackage{amssymb}
\usepackage{amsfonts}
\usepackage{mathtools}
\usepackage{cancel}
\usepackage{optidef}
\usepackage{xcolor}
\usepackage{bm}
\usepackage{booktabs}
\usepackage{longtable}
\usepackage{cancel}
\usepackage{tikz}
\usepackage{pgfplots}
\usepackage{caption}
\usepackage{subcaption}
\usepackage{adjustbox}
\usepackage{wrapfig}
\usepackage[normalem]{ulem}
\usepackage{multicol}
\usepackage{breqn}
\usepackage{supertabular}
\usepackage{resizegather}

\newcommand{\lpname}[1]{\text{\textup{\textsf{#1}}}}
\newcommand{\1}{\mathbf{1}}
\newcommand{\0}{\mathbf{0}}
\newcommand{\R}{\mathbb{R}}
\newcommand{\N}{\mathbb{N}}
\newcommand{\C}{\mathcal{C}}

\newcommand{\x}{X}
\newcommand{\y}{Y} 
\newcommand{\ve}{\varepsilon}
\newcommand{\Lin}{\textup{Lin}}

\newcommand{\cube}[1]{\{0,1\}^{#1}}
\newcommand{\DelsarteCube}[2]{\lpname{Delsarte}_{cube}(#1,#2)}
\newcommand{\Delsarte}[2]{\lpname{Delsarte}(#1,#2)}
\newcommand{\DelsarteLinCube}[3]{ \lpname{DelsarteLin}(#1,#2,#3)}
\newcommand{\DelsarteLinCubeNoParams}{ \lpname{DelsarteLin}}

\newcommand{\DelsarteLinCubeMod}[4]{ \lpname{DL}_{#1}(#2,#3,#4)}
\newcommand{\DelsarteLin}[3]{\lpname{DelsarteLin}_{/\mathfrak{S}_{#2}}(#1,#2,#3)}
\newcommand{\GL}[2]{\text{\textup{GL}}(#1,#2)}

\newtheorem{theorem}{Theorem}
\newtheorem{cnj}{Conjecture}

\newtheorem{definition}{Definition}
\newtheorem{lemma}{Lemma}
\newtheorem{proposition}{Proposition}

\DeclareMathOperator{\cf}{\Gamma}

\DeclareMathOperator{\maximize}{\textup{maximize}}
\DeclareMathOperator{\st}{\textup{subject~to:}}
\DeclareMathOperator{\val}{val}

\definecolor{color1}{rgb}{0.9725490196078431, 0.4627450980392157, 0.42745098039215684}
\definecolor{color2}{rgb}{0.0, 0.7294117647058823, 0.2196078431372549}
\definecolor{color3}{rgb}{0.3803921568627451, 0.611764705882353, 1.0}
\definecolor{color4}{rgb}{0.9019607843137255, 0.6235294117647059, 0.0}
\definecolor{color5}{rgb}{0.33725490196078434, 0.7058823529411765, 0.9137254901960784}

\title{New LP-based Upper Bounds in the Rate-vs.-Distance Problem for Binary Linear Codes}
\author{
    Elyassaf Loyfer\thanks{School of Computer Science and Engineering, Hebrew University, 91904 Jerusalem, Israel. Supported in part by grant 659/18 “High-dimensional combinatorics” of the Israel Science Foundation.}
    \and
    Nati Linial\footnotemark[1]
}

\begin{document}
	
\maketitle

\begin{abstract}
We develop a new family of linear programs, that yield upper bounds
on the rate of binary linear codes of given distance.
Our bounds apply {\em only to
linear codes.} Delsarte's LP is the weakest member of this family
and our LP yields increasingly tighter upper bounds on the rate as its control parameter
increases. Numerical experiments show significant improvement
compared to Delsarte. These convincing numerical results, and the
large variety of tools available for asymptotic analysis,
give us hope that
our work will lead to new improved asymptotic upper bounds
on the possible rate of linear codes.

A concurrent work by Coregliano, Jeronimo and Jones
offers a closely related family of linear programs which converges 
to the true bound. Here we provide a new proof of convergence for the same LPs.
\end{abstract}

\section{Introduction}\label{section:introduction}

The {\em rate vs.\ distance} problem is a fundamental question in coding
theory. Namely, we seek the largest possible rate of a code of given minimal
distance. The best lower bound that we have is due to Gilbert
\cite{gilbert1952comparison} (for general codes) and Varshamov
\cite{varshamov1957estimate} (in the linear case) and is attained by random
codes. The best upper bounds that we have are due to
McEliece, Rodemich, Ramsey and Welch
(MRRW) \cite{mceliece1977new}.
Based on Delsarte's
linear program \cite{delsarte1973algebraic},
these bounds are 
often called the first and second linear programming bounds.
There is substantial empirical evidence \cite{barg1999numerical} indicating that
the MRRW bounds may be asymptotically all that Delsarte's LP yields.

We propose a new family of linear programs, which greatly strengthen Delsarte's
LP. We stress that these new LPs apply {\em only} to linear codes. They come
with a control parameter, an integer $r$. For $r=1$ our LP coincides with
Delsarte's, and as $r$ increases the LP yields tighter upper
bounds on the code's rate, at the cost of higher complexity. Numerical
experiments (Figure \ref{fig:results_diff_plot}) show that even
with $r=2$ our LP is far stronger than Delsarte's,
surpassing it in almost all instances. The improved upper bound on the code's
size is up to $2.5$ times smaller than what Delsarte gives. Moreover, in all
instances where Delsarte's upper bound is known not to be tight, we improve it.
Nevertheless, our results do not improve the best known upper bounds.

\begin{figure}[t]
\centering
\makeatletter
\@for\dval:={6,8,10,12,14,16}\do{
\begin{subfigure}[b]{0.3\textwidth}
\centering
\begin{adjustbox}{width=\linewidth}
\begin{tikzpicture}
\begin{axis}[
    legend style={at={(0,1.3)},anchor=north west,legend cell align=left,font=\large},
    grid=major,
]
    \addplot[ultra thick,color=color1, mark=o]
        table [x=n, y=delsarte, col sep=comma] {data/diff_plot_table_d_\dval.csv};
    \ifnum \dval=6 {\addlegendentry{Delsarte}} \else {} \fi
    \addplot[ultra thick,color=color3, mark=square]
        table [x=n, y=strong_lp, col sep=comma] {data/diff_plot_table_d_\dval.csv};
    \ifnum \dval=6 {\addlegendentry{$\DelsarteLinCube{2}{n}{d}$}} \else {} \fi
    \addplot[ultra thick,color=color2, mark=x]
        table [x=n, y=weak_lp, col sep=comma] {data/diff_plot_table_d_\dval.csv};
    \ifnum \dval=6 {\addlegendentry{ \textup{KrawtchoukLin}$(n,d,2)$ \cite{coregliano2021complete}}} \else {} \fi
\end{axis}
\end{tikzpicture}
\end{adjustbox}
\caption{$d=\dval$}
\end{subfigure}
}
\makeatother
\caption{
(Lower is better). 
Numerical results comparing 
the bounds obtained by
optimization problems,
and the currently best known upper bounds, as reported in \cite{Grassl:codetables}.
Each line is the ratio between the optimal value of the corresponding LP and
the best known upper bound on $A^{\Lin}(n,d)$,
scaled by $\log_2(\cdot)$.
The $x$-axis in each plot varies over $n$.
Experimental setup, as well as more detailed results,
are given in Appendix \ref{section:numerical_results}.
For the bounds stated in terms of the codes' {\em dimension},
see Figure \ref{fig:results_diff_plot_floor}.}

\label{fig:results_diff_plot}
\end{figure}
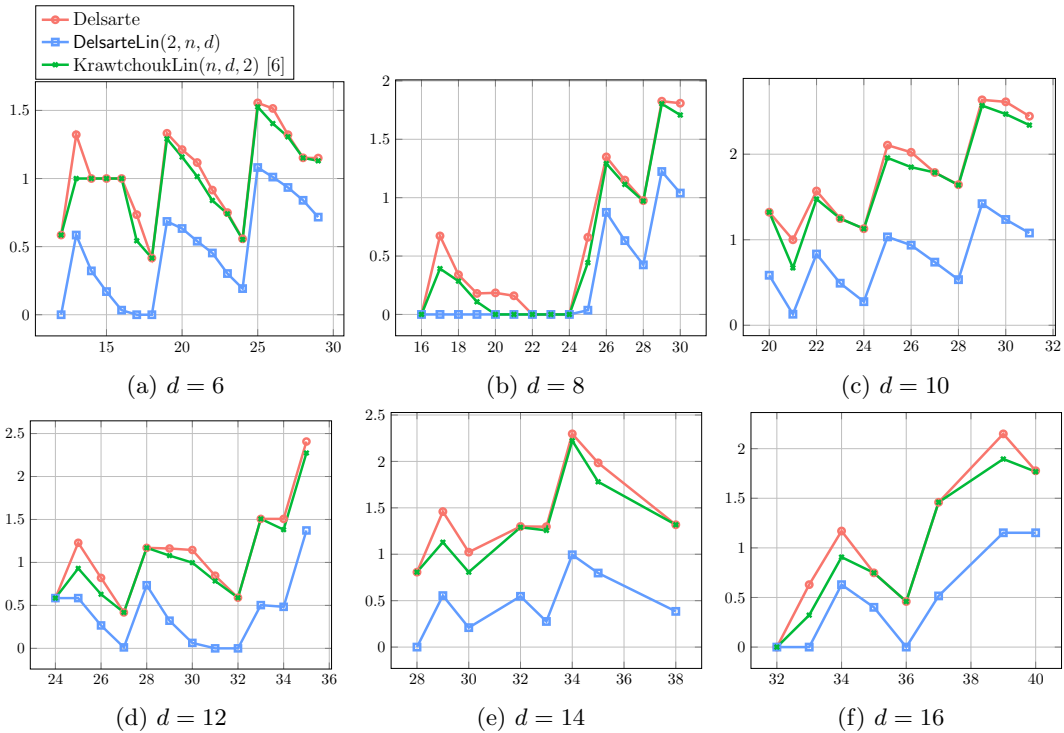

Our construction is based on the elementary fact
that a linear code is closed under addition.
Combined with Delsarte's LP this simple fact has considerable consequences. To
actually derive them we use (i) The language of Boolean Fourier analysis
(ii) Symmetry that is inherent in the problem. In analyzing Delsarte's LP,
symmetrization reduces the problem size from exponential to polynomial in $n$,
and brings Krawtchouk polynomials to the fore. Also here does symmetrization
yield a dramatic reduction in size and
reveals the role of multivariate Krawtchouk polynomials.
There is a large body of work on these
high-dimensional counterparts of univariate Krawtchouks,
e.g., \cite{diaconis2014introduction,grunbaum2011system}.

Although we are still unable to reach our main goal and derive better {\em
asymptotic} bounds, there is good reason for hope. Over forty years since it was proved,
the first MRRW bound is still the best upper bound that we have for a large
range of parameters. Over the years this bound has
been reproved using various tools and techniques.
These include, properties of Krawtchouk
polynomials \cite{mceliece1977new}, analysis of Boolean functions
(e.g., \cite{friedman2005generalized,navon2005delsarte,navon2009linear,
samorodnitsky2021one}), and
spectral \cite{barg2006spectral}
as well as functional \cite{barg2008functional} analysis.
We believe that
it is a viable and promising direction to
extend proofs of the first MRRW bound to our multivariate
LP family. We are hopeful that this will lead to stronger bounds
on the rate of linear codes. We focus here on binary codes, however 
our methods can be extended to $q$-ary codes as well.

\subsection{Related Work}

A concurrent work by Coregliano, Jeronimo and Jones
\cite{coregliano2021complete} 
employs closely related ideas
to produce a family of linear programs, which
upper bound the size of linear codes.
In comparison, our LP is stricter due to several conceptual new ideas
that we introduce here. 
Numerical comparisons between our LP and that of \cite{coregliano2021complete}
appear in Figure \ref{fig:results_diff_plot} and Appendix \ref{section:numerical_results}.
We indicate the
differences throughout the text where appropriate, in particular in Sections
\ref{subsec:C2} and \ref{subsec:objective_function}.

In \cite{coregliano2021complete}, they suggest two semi-definite
programs (SDPs) which are equivalent to the LP family. 
One SDP is then used to prove that their program converges
to the true bound
as the control parameter grows. Here we suggest an alternative 
proof, which bypasses the use of SDPs.

It was Schrijver \cite{schrijver2005new} who suggested
to find an SDP that strengthens Delsarte's LP.
His SDP improved the best upper bound
for general codes in several finite instances, but there
is still no known method to improve the asymptotic bounds using this SDP.
Our LP yields tighter bounds than those of \cite{schrijver2005new}, see
Appendix \ref{section:numerical_results}.

\subsection{Organization of this Paper}

The rest of the work is organized as follows. Section
\ref{section:preliminaries} provides preliminaries and notation. In Section
\ref{section:new_lps} we develop our new LP family and discuss some
of its properties. In particular, we
prove its strength, examine its components, and suggest some variations that may prove
useful in the asymptotic analysis. In Section \ref{subsec:approx_completeness} we
provide an alternative proof for the convergence theorem of 
\cite{coregliano2021complete}.
In Section \ref{section:symmetrized_lps} we derive the symmetrized LP. 

The derivation of our LPs motivates the definition of a new linear operator
which we call {\em partial Fourier transform}.
In Section \ref{section:partial_fourier_transform_symmetries} we explore some
of its characteristics which are relevant to our LP. The main result
of this section is an interesting equivalence between two properties
of the code's indicator function.

Section \ref{section:Krawtchouk_properties} connects our construction to the
literature on
multivariate Krawtchouk polynomials. These polynomials appear naturally
when we symmetrize the LP. In addition,
we develop the {\em partial} multivariate Krawtchouks, 
which are derived from the symmetrization of partial Fourier
transform.

Appendix \ref{section:numerical_results}
shows results from numerical experiments on a wide range of parameters.

\section{Notation and Preliminaries}\label{section:preliminaries}
\subsection{General}
We denote by $\N$ the set of nonnegative integers. For a positive integer $r$,
$[r]\coloneqq \{1,2,\dots,r\}$.
Vectors are distinguished from scalars by boldface letters,
e.g. $\bm{x}=(x_1,x_2,\dots,x_n)\in\R^n$.

We consider two linear programs {\em
equivalent} if their respective optimal values are equal. Likewise, relations
between LPs, e.g. "$=$","$\leq$", refer to optimal values. We denote by $\val
(\cdot)$ the optimal value of an LP.

\subsection{The Boolean Hypercube}
The $n$-dimensional Boolean hypercube, or simply \emph{the cube} is, as usual,
the linear space $\mathbb{F}_2^{n}$ or the set $\{0,1\}^{n}$.
An element, or a vector, in the cube is denoted in bold, e.g.
$\bm{x}=(x_1,\dots,x_n)\in\cube{n}$.
Addition $\bm{x}+\bm{y}\in\cube{n}$ is bitwise "xor", or
element-wise sum modulo $2$.
Inner product between vectors in the cube is done over $\mathbb{F}_2$:
$\langle \bm{x},\bm{y}\rangle = \sum_{i=1}^{n}x_iy_i \bmod 2$.

Let $f,g: \cube{n}\to\R$ be two real functions on the cube. Their inner product is defined as
\[
    \langle f,g\rangle = 2^{-n}\sum_{\bm{x}\in\cube{n}}f(\bm{x})g(\bm{x})
\]
and their convolution 
\[
    (f*g)(\bm{x}) = 2^{-n}\sum_{\bm{y}\in\cube{n}}f(\bm{y})g(\bm{x}+\bm{y})
\]
The tensor product of $f$ and $g$ is a
function on $\cube{2n}$: \[(f\otimes g)(\bm{x},\bm{y}) = f(\bm{x})g(\bm{y}).\]


The Hamming weight of $\bm{x}$, denoted
$|\bm{x}|$, is the number of non-zero bits, $|\bm{x}| = |\{1\leq i\leq n :
x_i\neq 0\}|$. 
For $i=0,\dots,n$, the $i$-th \emph{level-set} is the set of all Boolean
vectors of weight $i$. The indicator of the $i$-th level-set is called $L_i$:
\[
    L_i(\bm{x}) =
    \begin{cases}
        1 & |\bm{x}| = i \\
        0 & \text{o/w}
    \end{cases},\quad
    \bm{x}\in\cube{n}
\]

We denote Kronecker's delta function by $\delta_{\bm{x}}(\bm{y})$.

\noindent
The \emph{Fourier character}
corresponding to $\bm{x}\in\cube{n}$, denoted
$\chi_{\bm{x}}$, is defined by
\[
\chi_{\bm{x}}({\bm{y}}) = (-1)^{\langle {\bm{x}},{\bm{y}}\rangle},\quad
{\bm{y}}\in\cube{n}.
\]
The set of characters $\{\chi_{\bm{x}}\}_{{\bm{x}}\in\cube{n}}$ is an orthonormal basis for
the space of real functions on the cube. The Fourier transform of a function
$f:\cube{n}\to\R$ is its projection over the characters, $\hat{f}({\bm{x}}) =
\langle f,\chi_{\bm{x}}\rangle =
2^{-n}\sum_{\bm{y}}\chi_{\bm{x}}({\bm{y}})f({\bm{y}})$. In Fourier space, the
inner product is not normalized: $\langle \hat{f},\hat{g}\rangle_{\mathcal{F}}
= \sum_{\bm{x}}\hat{f}({\bm{x}})\hat{g}({\bm{x}})$.

We recall Parseval's identity: $\langle f,g\rangle = \langle
\hat{f},\hat{g}\rangle$; and the convolution theorem: $(\widehat{f*g})({\bm{x}}) =
\hat{f}({\bm{x}})\hat{g}({\bm{x}})$.

The partial Fourier transform, denoted $\mathcal{F}_S$
that we introduce here plays an important role in our work, see Section \ref{section:new_lps} 
for details.

A comprehensive survey of harmonic analysis of Boolean functions can be found in
\cite{o2014analysis}.

\subsection{Codes}
A \emph{binary code} of length $n$ is a subset $\C \subset \{0,1\}^{n}$. 
Its \emph{distance} is the smallest
Hamming distance between pairs of words, $dist(\C) =
\min_{{\bm{x}},{\bm{y}}\in\C} |{\bm{x}}+{\bm{y}}|$. The largest cardinality of a code
of length $n$ and distance $d$ is denoted by $A(n,d)$. The \emph{rate} of $\C$
is defined as
\[
    R(\C) = n^{-1}\log_2(|\C|)
\]
The asymptotic rate-vs.-distance problem is to find, for every
$\delta\in(0,1/2)$,
\[
    \mathcal{R}(\delta)
    = \limsup_{n\to\infty}
    \{ R(\C) : \C\subset\{0,1\}^{n},~ dist(\C) \ge \delta n\}
\]
A linear code is a linear subspace. In the binary case, $\C\subset\{0,1\}^n$ is
linear if and only if ${\bm{x}},{\bm{y}}\in\C \Rightarrow
{\bm{x}}+{\bm{y}}\in\C$, for every ${\bm{x}},{\bm{y}}\in\{0,1\}^n$.
Consequently, in a linear code $dist(\C)=\min_{{\bm{0}\neq\bm{x}}\in\C}|{\bm{x}}|$.
We denote by $A^{\Lin}(n,d)$ the maximal size of a binary linear code
of length $n$ and minimal distance $d$.

\section{New Linear Programs}\label{section:new_lps} 

In this section we present a new family of linear programs,
starting from Delsarte's LP.
Later in this Section we discuss possible modifications
to the LPs.

Let $\C\subset\cube{n}$ be a code,
not necessarily linear, with minimal distance $d$. Let
$\1_\C$ be its indicator function, namely 
$\1_\C(\bm{x})=1$ if $\bm{x}\in\C$, and $0$ otherwise.
Define the function
\[
    f_\C = \frac{2^n}{|\C|} \1_\C * \1_\C
\]
As we explain shortly, Fourier analysis of $f_\C$ yields 
Delsarte's LP for binary codes.
Our new LP family is likewise obtained by
considering the tensor product of copies of $f_\C$.

Indeed, it is easily verified that $f_\C(0) = 1$, and $f_\C(\bm{x}) = 0$ whenever $1\leq |\bm{x}|\leq d-1$.
In addition, $f_\C \geq 0$
as a sum of indicator functions.
Also, $\hat{f}_\C \geq 0$
because, by the convolution theorem, it is a squared function: 
$\hat{f}_\C = \frac{2^n}{|C|}\hat{\1}_\C^2$.
Lastly, summing $f_\C$ over the entire cube yields the cardinality of $\C$.
This yields the following LP, whose optimal value is an upper bound on $A(n,d)$.
\begin{definition}\label{def:delsarte_cube}
    $\DelsarteCube{n}{d}$ is the following linear program:
    \begin{alignat*}{2}
        & \mathrlap{ \underset{{f:\cube{n}\to\R}}{\maximize}
            \quad \sum_{\bm{x}\in\cube{n}} f(\bm{x})
            \tag{$obj$}\label{eq:delsarte_cube:objective}
            }
            \\
        &\st \\
        &\quad f(0) = 1, \tag{$d1$}\label{eq:delsarte_cube:c1} \\
        &\quad f \geq 0,~ \hat{f} \geq 0, \quad \tag{$d2$}\label{eq:delsarte_cube:c2} \\
        &\quad f(\bm{x}) = 0 \quad & \text{\textup{if} } 1\leq |\bm{x}|\leq d-1
            \tag{$d3$}\label{eq:delsarte_cube:c3}
    \end{alignat*}
\end{definition}

Now let us assume further that $\C$ is linear. In this case,
$
    \1_\C(\bm{x})\1_\C(\bm{y}) = \1_\C(\bm{x})\1_\C(\bm{x}+\bm{y}),
$
and consequently,
\[
    f_\C = \frac{1}{|\C|}\1_\C*\1_\C = \1_\C.
\]
This implies a new set of constraints that hold for linear codes and can be added to the above LP:
\[
    f(\bm{x})f(\bm{y}) = f(\bm{x})f(\bm{x}+\bm{y})
\]
However, these constraints are not linear in $f$, nor even convex.

Therefore, we consider instead tensor products
of $f_\C$. Let $r\geq 1$ be an integer and define 
\[
    f_{\C^{r}} = f_\C\otimes \dots \otimes f_\C :\cube{rn}\to\R.
\]
The function $f_{\C^{r}}$ is defined on the $rn$-dimensional cube. 
We will view its argument as either a concatenation of $r$ vectors in $\cube{n}$,
or an $r\times n$ matrix obtained by stacking the $r$ vectors.
For example,
we write
\[
    f_{\C^{r}}(X)
    = f_{\C^{r}}(\bm{x}_1,\dots,\bm{x}_r)
\]
where $\bm{x}_1,\dots,\bm{x}_r\in\cube{n}$ are the rows of the matrix $X\in\cube{r\times n}$.

As suggested above, our LP family is derived from the linear
properties of 
$f_{\C^r}$. Some of these properties apply even for non-linear $\C$
and are inherited from the properties of the original $f_{\C}$.
The other type is properties that depend on the linearity of $\C$.
We turn to describe both types.

We begin with the first type. It is clear that
$f_{\C^r}(\0)=1$, and $f_{\C^r}(\bm{x}_1,\dots,\bm{x}_r)=0$ if
any of the vectors $\bm{x}_1,\dots,\bm{x}_r$ has weight between
$1$ and $d-1$.
The non-negativity of $f_\C$ and $\hat{f}_\C$ imply the same for
$f_{\C^r}$. But there is more: products of $f_\C$ and $\hat{f}_\C$
are also non-negative,
e.g.\ $f_{\C}(\bm{x}_1) \hat{f}_{\C}(\bm{x}_2)\geq 0$ for every
$\bm{x}_1,\bm{x}_2\in\cube{n}$.

This motivates the definition of a new linear operator, which we
name {\em partial Fourier transform}.
\begin{definition}\label{def:partial_fourier_transform}
    Let $S\subset [r]$ and $\bm{x}_1,\dots,\bm{x}_r \in\cube{n}$.
    The \textbf{partial Fourier character}
    $\Psi^{S}_{(\bm{x}_1,\dots,\bm{x}_r)}$ is defined by
    \[
        \Psi^{S}_{(\bm{x}_1,\dots,\bm{x}_r)}
        \coloneqq \psi^{(1)}_{\bm{x}_1} \otimes \psi^{(2)}_{\bm{x}_2} \otimes \dots \otimes \psi^{(r)}_{\bm{x}_r}
    \]
    where, given $\bm{x}\in\cube{n}$,
    \[
        \psi^{(i)}_{\bm{x}} \coloneqq
        \begin{cases}
            \chi_{\bm{x}} & i \in S \\
            \delta_{\bm{x}} & \text{o/w}
        \end{cases}
    \]
    $\chi_{\bm{x}}$ is a Fourier character in $\cube{n}$ and $\delta_{\bm{x}}$ is Kronecker's delta.
    
    The \textbf{partial Fourier transform} is the linear projection
    of a function \mbox{$g:\cube{rn}\to\R$} on the partial characters,
    \begin{dmath*}
        \mathcal{F}_S(g)(\bm{x}_1,\dots,\bm{x}_r)
        = 2^{(r-|S|)n}\langle g, \Psi^{S}_{(\bm{x}_1,\dots,\bm{x}_r)} \rangle
        = 2^{-|S|n} \sum g(\bm{y}_1,\dots,\bm{y}_r)
            \times
            \\
            \quad
            \times
            \prod_{i\in S} \chi_{\bm{x}_i}(\bm{y}_i)
            \prod_{i\in[r]\setminus S} \delta_{\bm{x}_i}(\bm{y}_i)
    \end{dmath*}
    the sum running over all $\bm{y}_1,\dots,\bm{y}_r\in \cube{n}$.
\end{definition}
Observe that $\mathcal{F}_\emptyset(g) = g$, $\mathcal{F}_{[r]}(g) = \hat{g}$, and
$\mathcal{F}_{\{i,j\}}(g) = \mathcal{F}_{\{i\}}(\mathcal{F}_{\{j\}}(g))$, for
$1\leq i,j\leq r$, $i\neq j$.

Using the new notation, we have $\mathcal{F}_S(f_{\C^r}) \geq 0$ for every
$S\subset [r]$.

The last inherited property of $f_{\C^r}$ has to do with the cardinality of $\C$. 
Summing $f_{\C^r}$ over the entire $rn$-dimensional cube
yields $|\C|^r$. Alternatively, one can obtain the value of $|\C|$
by summing one component over $\cube{n}$, and fixing the other
components at $0$:
\[
    \sum_{\bm{x}\in\cube{n}} f_{\C^r}(\bm{x},0,\dots,0)
    = \sum_{\bm{x}\in\cube{n}} f_{\C}(\bm{x}) \left(f_{\C}(0)\right)^{r-1}
    = |C|
\]

We turn to discuss the properties which depend on 
the linearity of the code $\C$.
If $\C$ is a linear code and $\bm{x}_1,\dots,\bm{x}_r\in \C$, then $\C$
contains their linear span. Hence
\[
    f_{\C^r}(X) = \prod_{i=1}^{r} \1_{\C}(\bm{x}_i) = \prod_{\bm{x}\in rowspan(X)} \1_\C(\bm{x})
\]
which implies that $f_{\C^r}$ is invariant under the action of $\GL{r}{2}$, the general linear group
over $\mathbb{F}_2$.
\begin{equation}
    \label{eq:fCr_invariant_under_GL}
    f_{\C^r}(X) = f_{\C^r}(TX) \quad \forall T\in \GL{r}{2},~ X\in\cube{r\times n}
\end{equation}

One more interesting property involves the dual code,
\[
    \C^\perp \coloneqq \{\bm{x} \in \cube{n} : \langle \bm{x},\bm{y} \rangle_{\mathbb{F}_2}=0~ \forall \bm{y}\in\C\}.
\]
If $\C$ is linear, then the Fourier transform of its indicator
$\hat{\1}_\C$ is the indicator of the dual code,
up to normalization (see e.g., \cite{o2014analysis}, Proposition 3.11):
\[
    \hat{\1}_{\C} = \frac{1}{\left|\C^\perp\right|} \1_{\C^\perp}
\]
This fact can be utilized through the partial Fourier transform as follows.
Let $\bm{x},\bm{y}\in\cube{n}$
such that $\langle\bm{x},\bm{y}\rangle_{\mathbb{F}_2} \neq 0$,
then either $\bm{x}\notin\C^\perp$ or $\bm{y}\notin\C$. Consequently,
if $i\in S$ and $j\notin S$ for some $S\subseteq [r]$, and 
$\langle\bm{x}_i,\bm{x}_j\rangle_{\mathbb{F}_2} \neq 0$, then
\begin{equation}
    \label{eq:fCr_dual_code}
    \mathcal{F}_{S}(f_{\C^r})(\bm{x}_1,\dots,\bm{x}_r)
    = \left(\prod_{k\in S} \frac{1}{|\C^\perp|} \1_{\C^\perp} (\bm{x}_k)\right)
    \left(\prod_{k\in [r]\setminus S} \1_{\C}(\bm{x}_k)\right)
    =0
\end{equation}
Surprisingly perhaps, this adds no new information:
properties \eqref{eq:fCr_invariant_under_GL} and \eqref{eq:fCr_dual_code}
are equivalent, as we show
in Section
\ref{section:partial_fourier_transform_symmetries}.

This concludes our discussion on the linear properties
of the tensor product $f_{\C^r}$. We are now ready to
define the new LP family.

\begin{definition}\label{def:delsarte_lin_cube}
    $\DelsarteLinCube{r}{n}{d}$:
    \begin{alignat*}{3}
        & \mathrlap{\underset{f:\cube{rn}\to\R}{\maximize}\quad
            \sum_{\bm{x}\in\cube{n}} f(\bm{x},0,\dots,0)
            }
            \tag{$Obj$}\label{eq:delsarte_lin:objective}
        \\
        &\st \\
        &\quad f(\mathbf{0}) = 1 \tag{$C1$}\label{eq:delsarte_lin:C1}
        \\
        &\quad \mathcal{F}_S(f) \geq 0 
            \quad
            && \forall S\subset [r]
            \tag{$C2$}\label{eq:delsarte_lin:C2}
        \\
        &\quad f(\bm{x}_1,\dots,\bm{x}_r) = 0 
            \quad
            &&
            \text{if }  1\leq |\bm{x}_1|\leq d-1 
            \tag{$C3$}\label{eq:delsarte_lin:C3}
        \\
        &\quad f(\x) = f(T\x)
            \quad 
            &&
            \forall T\in \GL{r}{2},~X\in\cube{r\times n}
            \tag{$C4$}\label{eq:delsarte_lin:C4}
    \end{alignat*}
\end{definition}
Here, $\mathcal{F}_S(f)$ is the partial Fourier transform defined above.
Also,
$\GL{r}{2}$ is the general linear group over $\mathbb{F}_2$. Note also the parallels between conditions 
\eqref{eq:delsarte_cube:c1}, \eqref{eq:delsarte_cube:c2}, \eqref{eq:delsarte_cube:c3}
resp.\
\eqref{eq:delsarte_lin:C1}, \eqref{eq:delsarte_lin:C2}, \eqref{eq:delsarte_lin:C3}

\vspace{5mm}

\begin{theorem}\label{thm:lp_strength}
    Let $r,n,d$ be positive integers such that $d\leq n/2$.
    \begin{enumerate}
        \item
        $
            A^{\Lin}(n,d) \leq \val\DelsarteLinCube{r}{n}{d} 
        $
        \item
        $
            \val\DelsarteLinCube{r+1}{n}{d}
            \leq \val\DelsarteLinCube{r}{n}{d} 
        $
        \item
        $
           \val\DelsarteLinCube{1}{n}{d} 
            = \val\DelsarteCube{n}{d} 
        $
    \end{enumerate}
\end{theorem}

We make a few comments before we turn to the proof. Already for $r=2$, and in most instances,
$\DelsarteLinCubeNoParams$ is significantly stronger than
Delsarte's. For more on this, see
Figure \ref{fig:results_diff_plot} and Section
\ref{section:numerical_results}. We also note that $\DelsarteLinCubeNoParams$
without \eqref{eq:delsarte_lin:C4} yields exactly the bounds as
Delsarte's LP.

\begin{proof}\hfill

\begin{enumerate}
    \item By the preceding discussion, for every binary linear code
    $\C$ of length $n$ and minimal distance $d$, $f_{\C^r}$ is a feasible solution with 
    value $|\C|$.
    
    \item Let $f:\cube{(r+1)n}\to\R$ be a 
    feasible solution to $\DelsarteLinCube{r+1}{n}{d}$. We construct
    a feasible solution to $\DelsarteLinCube{r}{n}{d}$ with value
    at least $\val(f)$.
    
    Let
    \[
        g:\cube{rn}\to\R,\quad g(\bm{x}_1,\dots,\bm{x}_r)
        = f(\bm{x}_1,\dots,\bm{x}_r,\bm{0})
    \]
    It is easy to verify that $g$ is feasible for 
    $\DelsarteLinCube{r}{n}{d}$, and it is clear that
    $\val(g)=\val(f)$.
    
    \item Obvious, 
    $\DelsarteLinCube{1}{n}{d}$ and 
    $\DelsarteCube{n}{d}$ are identical.
\end{enumerate}
\end{proof}

In the rest of this section, we examine the strength and consequences
of some components of $\DelsarteLinCubeNoParams$.
We also discuss two modifications that
may be helpful in the search for asymptotic results.

\subsection{On the significance of \eqref{eq:delsarte_lin:C4}}
\label{subsec:C4}

As mentioned above, \eqref{eq:delsarte_lin:C4} is equivalent to
a constraint that uses the dual code:
\begin{align*}
    \mathcal{F}_S(f)(\bm{x}_1,\dots,\bm{x}_r) = 0
    \quad &\text{ if } \langle\bm{x}_i, \bm{x}_j\rangle_{\mathbb{F}_2} = 1
    \\
    & \text{ for some } i\in S,~j\notin S
    \tag{$C5$}\label{eq:delsarte_lin:C5}
\end{align*}
We prove the equivalence
below, in Lemma \ref{lemma:C456_equivalence}.
As \eqref{eq:delsarte_lin:C4} and \eqref{eq:delsarte_lin:C5} are
the only constraints that rely on the code's linearity,
without them the
LP is equivalent to Delsarte's LP, for every $r$.

\vspace{5mm}

An obvious consequence of \eqref{eq:delsarte_lin:C4} is
that \eqref{eq:delsarte_lin:C3} is equivalent to
\[
    f(X) = 0~ \text{if } 1\leq |\bm{u}^ \intercal X | \leq d-1
    \text{ for some } \bm{u}\in\cube{r}
    \tag{$C3'$}\label{eq:delsarte_lin:C3_span}
\]
$\{\bm{u}^\intercal X\}_{\bm{u}\in\cube{r}}$ is the row span of $X$.
Similarly, \eqref{eq:delsarte_lin:C4}
renders the objective \eqref{eq:delsarte_lin:objective} equivalent to
\[
    \maximize \quad (2^r-1)^{-1}
        \sum_{0\neq \bm{u}\in\cube{r}}\sum_{\bm{x}\in\cube{n}}
            f(u_1 \cdot \bm{x},\dots,u_r\cdot \bm{x})
    \tag{$Obj''$}\label{eq:delsarte_lin:objective_span}
\]

To numerically test the significance of \eqref{eq:delsarte_lin:C4},
we removed it but
kept its immediate consequences. Namely, we
replaced
\eqref{eq:delsarte_lin:objective} and \eqref{eq:delsarte_lin:C3}
with 
\eqref{eq:delsarte_lin:objective_span}
and \eqref{eq:delsarte_lin:C3_span}.
A sample from our
numerical experiments is shown in Figure
\ref{fig:C4_significance}. It confirms that 
this change does weaken the LP, though 
not significantly. However, we only experimented with $r=2$,
and it is possible that for larger values of $r$
the difference becomes more substantial.
\begin{figure}[h]
\centering
\begin{tabular}{cc|rrr}
      & r & 1 & 2 & 2 \\  
    n & d & Delsarte & $\DelsarteLinCubeNoParams$ & 
        ${
            \eqref{eq:delsarte_lin:objective_span},
            \eqref{eq:delsarte_lin:C3_span},
            \cancel{\eqref{eq:delsarte_lin:C4}}}$ \\
    \toprule
    16 & 4 & 2048 & 2048 & 2048 \\
    & 6 & 256 & 131.72 & 156.44 \\
    & 8 & 32  & 32      & 32 \\
    \midrule
    17 & 4 & 3640.89 & 3072.96 & 3075 \\
    & 6 & 425.56 & 256 & 264.88 \\
    & 8 & 50.72 & 32 & 32.31
\end{tabular}
\caption{
Numerical experiments on the significance of
\eqref{eq:delsarte_lin:C4}.
The first column is Delsarte's LP.
The second column is $\DelsarteLinCube{2}{n}{d}$.
The third column is a modification of $\DelsarteLinCube{2}{n}{d}$,
where
\eqref{eq:delsarte_lin:C3} is replaced by
\eqref{eq:delsarte_lin:C3_span};
the objective function is replaced by
\eqref{eq:delsarte_lin:objective_span};
and \eqref{eq:delsarte_lin:C4} is removed.}
\label{fig:C4_significance}
\end{figure}

Lastly, \eqref{eq:delsarte_lin:C4} implies other
symmetries for $\mathcal{F}_{S}(f)$. While these do not strengthen
the LP, they provide an exponential in $r$ reduction in the number of constraints.
For proof, see Lemma \ref{lemma:C456_equivalence}.
    \[
        \mathcal{F}_{S}(f)(X) = \mathcal{F}_{S}(f)(T_1 T_2 X),
        \tag{$C6$}\label{eq:delsarte_lin:C6}
    \]
    for every $T_1,T_2 \in \GL{r}{2}$, such that 
    $T_1 e_i = e_i ~ \forall i\in S$ and 
    $T_2 e_i = e_i ~ \forall i\in [r]\setminus S$.
    
    \[
        \mathcal{F}_{S}(f)(\bm{x}_1,\dots,\bm{x}_r)
        = \mathcal{F}_{\pi^{-1}(S)}(f)(
            \bm{x}_{\pi(1)},\dots,\bm{x}_{\pi(r)}),
        \tag{$C7$}\label{eq:delsarte_lin:C7}
    \]
    for every $\pi\in\mathfrak{S}_r$ -- permutation on $r$
    elements.

\subsection{On the significance of \eqref{eq:delsarte_lin:C2}}
\label{subsec:C2}

A weaker, simpler LP is obtained from $\DelsarteLinCube{r}{n}{d}$
by replacing \eqref{eq:delsarte_lin:C2} with
\[
    f \geq 0,~ \hat{f} \geq 0
    \tag{$C2'$}\label{eq:delsarte_lin:C2_weak}
\]
This modification restores the feasible region of
the LP developed by \cite{coregliano2021complete}.
The modified LP is still stronger than Delsarte's LP, and it becomes stronger 
with growing $r$, as we prove in Theorem \ref{thm:weak_lp_strength}.
Its simplicity might make it more suitable for asymptotic analysis.

We observed empirically that this modification
greatly weakens the LP. A small sample is given here in Figure
\ref{fig:compare_C2_weak}, and
more can be found in Section \ref{section:numerical_results}
and in Figure \ref{fig:results_diff_plot}.
\begin{figure}[h]
\centering
\begin{tabular}{cc|rrrr}
      & r & 1 & 2 & 2 & 3 \\
    n & d & Delsarte & 
        \eqref{eq:delsarte_lin:C2} &
            \eqref{eq:delsarte_lin:C2_weak} &
                \eqref{eq:delsarte_lin:C2_weak} \\
    \toprule
    13 & 6 & 40 & 24.26 & 32 & 23.07 \\
    \midrule
    30 & 8 & 114816 & 71094.5 & 107044 & - \\
    30 & 10 & 12525.4 & 5928.52 & 11340.4 & - \\
    30 & 12 & 1131.79 & 582.09 & 1026.28 & - \\
    30 & 14 & 129.68 & 80.08 & 112 & - \\
\end{tabular}
\caption{Comparison between \eqref{eq:delsarte_lin:C2},
\eqref{eq:delsarte_lin:C2_weak} and Delsarte.
Each column shows the optimal value a of different LP.
The LPs from left to right:
Delsarte's LP;
Our LP with $r=2$;
Our modified LP with \eqref{eq:delsarte_lin:C2_weak}
instead of \eqref{eq:delsarte_lin:C2}, with $r=2$;
and again the modified LP, with $r=3$.
This exhausts the results that we have for $r=3$.}

\label{fig:compare_C2_weak}
\end{figure}

\begin{theorem} \label{thm:weak_lp_strength}
    Let $r,n,d$ be positive integers such that $d\leq n/2$.
    For every binary linear code $\C$ with length $n$ and distance $d$,
    \[
        |\C|
        \leq \val \DelsarteLinCubeMod{\eqref{eq:delsarte_lin:C2_weak}}{r+1}{n}{d}
        \leq \val \DelsarteLinCubeMod{\eqref{eq:delsarte_lin:C2_weak}}{r}{n}{d}
    \]
    where $\DelsarteLinCubeMod{\eqref{eq:delsarte_lin:C2_weak}}{r}{n}{d}$ is
    the variant of $\DelsarteLinCube{r}{n}{d}$ in which \eqref{eq:delsarte_lin:C2}
    is replaced by \eqref{eq:delsarte_lin:C2_weak}.
\end{theorem}
\begin{proof}
    The first inequality follows from Theorem \ref{thm:lp_strength},
    by noting that every feasible solution
    to $\DelsarteLinCube{r}{n}{d}$ is a feasible solution
    to the modified version.
    
    For the second inequality, let $f:\cube{(r+1)n}\to\R$ be a feasible
    solution to $\DelsarteLinCubeMod{\eqref{eq:delsarte_lin:C2_weak}}{r+1}{n}{d}$.
    Define
    \[
        g:\cube{rn}\to\R,\quad
        g(\bm{x}_1,\dots,\bm{x}_r) =
            f(\bm{x}_1,\dots,\bm{x}_r,\bm{0})
    \]
    It is obvious that $g(\bm{0}) = 1$; $g \geq 0$;
    $g(\bm{x}_1,\dots,\bm{x}_r)=0$ if $1\leq |\bm{x}_1|\leq d-1$; and
    that
    $g(X) = g(TX)$ for every $T\in \GL{r}{2}$.
    To prove that $g$ is feasible, it
    remains to show that
    $\hat{g}\geq 0$. Observe that 
    \begin{dmath*}
        \hat{g}(\bm{x}_1,\dots,\bm{x}_r)
        = \mathcal{F}_{\{1,\dots,r\}}(f)(\bm{x}_1,\dots,\bm{x}_r,\bm{0})
        = 2^n \mathcal{F}_{\{r+1\}}(\hat{f})(\bm{x}_1,\dots,\bm{x}_r,\bm{0})
        = \sum_{\bm{y}\in\cube{n}}
            \chi_{\bm{0}}(\bm{y})\hat{f}(\bm{x}_1,\dots,\bm{x}_r,\bm{y})
        = \sum_{\bm{y}\in\cube{n}} \hat{f}(\bm{x}_1,\dots,\bm{x}_r,\bm{y})
    \end{dmath*}
    which is non-negative since $\hat{f}\geq 0$.
    The value of $f$ equals the value of $g$, which is at most
    $\val \DelsarteLinCubeMod{\eqref{eq:delsarte_lin:C2_weak}}{r}{n}{d}$.

\end{proof}

\subsection{On the objective function}
\label{subsec:objective_function}

As discussed above,
an alternative objective function can be used, which bounds $\left(A^\Lin(n,d)\right)^r$
instead of $A^\Lin(n,d)$:
\[
    \maximize
    \sum_{\bm{x}_1,\dots,\bm{x}_r\in\cube{n}}f(\bm{x}_1,\dots,\bm{x}_r)
    \tag{$Obj'$}\label{eq:delsarte_lin:objective_r}
\]
This is the objective function used in \cite{coregliano2021complete}.

Our numerical calculations reveal rather minor differences between the two
objectives, with no consistent advantage to one over the other. See
detailed results in Section \ref{section:numerical_results}. 

We state:
\begin{cnj}
    \label{cnj:alt_obj_strength}
    Let $r,n,d$ be positive integers such that $d\leq n/2$. Then\\
    \begin{gather*}
        \left(\val \DelsarteLinCubeMod{
            \eqref{eq:delsarte_lin:objective_r}}{r+1}{n}{d}\right)^{1/(r+1)}
        \leq \left(\val \DelsarteLinCubeMod{
                \eqref{eq:delsarte_lin:objective_r}}{r}{n}{d}\right)^{1/r}
    \end{gather*}
    Here, $\DelsarteLinCubeMod{\eqref{eq:delsarte_lin:objective_r}}{r}{n}{d}$
    is obtained from $\DelsarteLinCube{r}{n}{d}$ by replacing the objective
    function with \eqref{eq:delsarte_lin:objective_r}.
\end{cnj}

Due to the non-linear relation between the two objective functions
we are presently only able to prove the following.
A similar Theorem can likewise be proved 
for the variant where 
\eqref{eq:delsarte_lin:C2_weak} replaces \eqref{eq:delsarte_lin:C2}.
\begin{theorem}\label{thm:alt_obj_strength}
    Let $r,n,d$ be positive integers such that $d\leq n/2$. Then
    \begin{dmath*}
        \left(\val \DelsarteLinCubeMod{
            \eqref{eq:delsarte_lin:objective_r}}{r+1}{n}{d}\right)^{1/(r+1)}
        \leq 
        \max
        \begin{cases}
            \left(\val \DelsarteLinCubeMod{
                \eqref{eq:delsarte_lin:objective_r}}{r}{n}{d}\right)^{1/r},
            \\
            \val \DelsarteLinCube{r+1}{n}{d}
        \end{cases}
    \end{dmath*}
\end{theorem}

\begin{proof}
    Let $f:\cube{(r+1)n}\to\R$ be a feasible solution to
    $\DelsarteLinCubeMod{\eqref{eq:delsarte_lin:objective_r}}{r+1}{n}{d}$. Then
    $f$ is also a feasible solution to $\DelsarteLinCube{r+1}{n}{d}$.
    Let 
    \begin{align*}
        v_1 &= \left(\sum f(\bm{x}_1,\dots,\bm{x}_{r+1})\right)^{1/(r+1)}
    \\
        v_2 &= \sum f(\bm{x},\bm{0},\dots,\bm{0})
    \end{align*}
    where the sums are over $\bm{x}_1,\dots,\bm{x}_{r+1}\in\cube{n}$ and
    over $\bm{x}\in\cube{n}$, respectively.
    
    If $v_1\leq v_2$
    then we are done, because $v_2$ is not greater than
    the optimum of $\DelsarteLinCube{r+1}{n}{d}$.
    
    Otherwise, $v_1 >
    v_2$. Define
    $g:\cube{rn}\to\R$ as
    \[
        g(\bm{x}_1,\dots,\bm{x}_r)
            = \frac{1}{v_2}
            \sum_{\bm{y}\in\cube{n}}f(\bm{x}_1,\dots,\bm{x}_r,\bm{y})
    \]
    It is not hard to verify that $g$ is a feasible solution to
    $\DelsarteLinCubeMod{\eqref{eq:delsarte_lin:objective_r}}{r}{n}{d}$.
    Now consider its value:
    \[
        \left(\sum g(\bm{x}_1,\dots,\bm{x}_{r})\right)^{1/r}
        = \left(\frac{v_1^{r+1}}{v_2}\right)^{1/r}
        \\
        \geq v_1
    \]
    and the value of $g$ is at most the optimal value
    of $\DelsarteLinCubeMod{\eqref{eq:delsarte_lin:objective_r}}{r}{n}{d}$.
\end{proof}

\subsection{Approximate Completeness}
\label{subsec:approx_completeness}

Coregliano et.\ al.\ \cite{coregliano2021complete} prove that for $r$ large enough, the
LP family with the objective function \eqref{eq:delsarte_lin:objective_r}
converges to $A^{\Lin}(n,d)^r$. For binary linear codes, it can be stated as follows:
\begin{theorem}[Approximate Completeness]
    \label{thm:approx_completeness}
    Let $\ve \in (0,1)$ and $r \geq 2n^2/\log_2(1+\ve)$. Then
    \[
        \left(\val \DelsarteLinCubeMod{\eqref{eq:delsarte_lin:objective_r}}{r}{n}{d}\right)^{1/r}
        \leq (1+\ve) A^{\Lin}(n,d)
    \]
\end{theorem}
The proof in \cite{coregliano2021complete} is based on an SDP formulation which 
is equivalent to the LP family. 
The idea of the proof is to upper-bound the variables, 
and then count the non-zero variables. 
Our proof follows the same idea, without using an SDP.
The following proposition provides upper bounds on the variables,
which is followed by a count of the non-zero variables.
\begin{proposition}
    \label{prop:variables_upper_bound}
    Let $f:\cube{n}\to\R$ such that $f(0) = 1$ and $\hat{f}\geq 0$. Then $f\leq 1$.
\end{proposition}
\begin{proof}
    Let $0\neq \bm{x}\in \cube{n}$. Since $\hat{f}\geq 0$, we have
    \[
        0 \leq \sum_{\bm{y}: \langle \bm{y}, \bm{x} \rangle_{\mathbb{F}_2} = 1} \hat{f}(\bm{y})
        = \sum_{\bm{y}: \langle \bm{y}, \bm{x} \rangle_{\mathbb{F}_2} = 1} 
            \sum_{\bm{z}\in\cube{n}} \chi_{\bm{y}}(\bm{z}) f(\bm{z})
    \]
    For every $\bm{y}$ in the sum, there holds $\chi_{\bm{y}}(\bm{x}) = -1$ and $\chi_{\bm{y}}(0)=1$. Hence,
    \[
        0 \leq 2^{n-1}f(0) - 2^{n-1}f(\bm{x}) + 
        \sum_{\bm{z}\neq 0,\bm{x}} f(\bm{z})
        \sum_{\bm{y}: \langle \bm{y}, \bm{x} \rangle_{\mathbb{F}_2} = 1}\chi_{\bm{y}}(\bm{z})
    \]
    We complete the proof by showing that the last term vanishes.
    So, let $[\bm{x},\bm{z}]$ be the $2\times n$ matrix whose rows are $\bm{x}$ and $\bm{z}$. 
    The action of multiplying $[\bm{x},\bm{z}]$ by $\bm{y}\in\cube{n}$ divides the $n$-dimensional
    cube into cosets in $\mathbb{F}_2^2$. If $\bm{x}\neq \bm{z}$ and both are non-zero,
    then each coset has cardinality $2^{n-2}$.
    The inner sum is over the cosets $(1,0)$ and $(1,1)$. If $\bm{y}$ is in the first
    coset, then $\chi_{\bm{z}}(\bm{y})=1$, and if it is in the second then
    $\chi_{\bm{z}}(\bm{y})=-1$. In total, the sum vanishes.
    
\end{proof}
\begin{proof}[Proof of Theorem \ref{thm:approx_completeness}]
    Let $f$ be a solution to
    $\DelsarteLinCubeMod{\eqref{eq:delsarte_lin:objective_r}}{r}{n}{d}$.
    By Proposition \ref{prop:variables_upper_bound}, $f\leq 1$.
     
    Let $k_0$ be the largest possible dimension of a binary linear code with length $n$ and distance $d$,
    namely $2^{k_0} = A^{\Lin}(n,d)$. Then $f$ vanishes of every $r\times n$ binary matrix of $\mathbb{F}_2$-rank larger than $k_0$.
    Then, the value of the LP corresponding to $f$ is at most $\sum_{k=0}^{k_0}\gamma_{n,r,k}$,
    where $\gamma_{n,r,k}$ is the number of such matrices of rank exactly $k$.
    
    We next derive an upper bound on $\gamma_{n,r,k}$. There are exactly \mbox{$\prod_{i=1}^{k}(2^n-2^{i})\le 2^{nk}$} ordered bases 
    of $k$-dimensional subspaces of $\mathbb{F}^n_2$.
    There are $r(r-1)\cdots(r-k+1)\le r^k$ possible ways to place the chosen ordered
    base in an $r\times n$ matrix, and then $2^{k}$ options to choose each of the remaining rows without increasing the rank.
    Hence, $\gamma_k \leq 2^{nk} r^{k} 2^{k(r-k)}$, and
    \begin{align*}
        \val \DelsarteLinCubeMod{\eqref{eq:delsarte_lin:objective_r}}{r}{n}{d}
        &\leq \sum_{k=0}^{k_0}  2^{nk} r^{k} 2^{kr}
        \\
        &\leq (k_0+1) 2^{k_0(n + r + \log_2(r))}
        \\
        &\leq 2^{n^2 + n\log_2(r) + \log_2(n+1)} A^\Lin(n,d)^r
        \\
        &\leq (1+\ve)^{r} A^\Lin(n,d)^r
    \end{align*}
    in the last inequality we use the assumption that $r\geq 2n^2/\log_2(1+\ve)$.
\end{proof}

\section{Symmetrized Linear Programs}

\label{section:symmetrized_lps} 

Due to the inherent symmetries of the LPs from section \ref{section:new_lps}
they can be symmetrized without affecting the objective function.
The advantage is that the symmetrized LP is significantly smaller than the original form.
This is what we consider in this section.

Let $\mathfrak{S}_n$ be the symmetric group on $n$ elements. It
acts on $\cube{r\times n}$ by column permutations:
\[
    \sigma \cdot X = [ \bm{\xi}_{\sigma(1)},\dots, \bm{\xi}_{\sigma(n)} ]
\]
where $\bm{\xi}_1,\dots,\bm{\xi}_n$ are the columns of $X\in\cube{r\times n}$,
and $\sigma \in \mathfrak{S}_n$. It also acts on functions $f:\cube{r\times n}\to \R$
via $(\sigma\circ f)(X) = f(\sigma \cdot X)$. 

We say that a solution $f$ to $\DelsarteLinCube{r}{n}{d}$ is {\em symmetric}
if it is constant on $\mathfrak{S}_n$-orbits, i.e., if $f = \sigma \circ f$
for every $\sigma \in\mathfrak{S}_n$. Symmetric solutions can clearly be
described more concisely, and as we observe below, there exist optimal symmetric solutions.

Generally speaking, suppose that the group $G$ acts on the variables of 
a linear program $\mathcal{P}$. We say that $f$,
a feasible solution of $\mathcal{P}$
is $G$-{\em invariant} if
$g\circ f$ is feasible and $\val(g\circ f) = \val(f)$, for every $g\in G$.
If every  feasible solution is invariant,
we say that $\mathcal{P}$
is $G$-invariant.
An invariant solution $f$ need not be symmetric, but averaging can yield 
a symmetric solution via
\[
    \overline{f}\coloneqq |G|^{-1} \sum_{g\in G} g\circ f
\]
By linearity and convexity, $\overline{f}$ is feasible and has the same value as $f$.
Consequently, a $G$-invariant LP has a symmetric optimal solution.

Let us verify that $\DelsarteLinCube{r}{n}{d}$ is $\mathfrak{S}_n$-invariant.
Let $f$ be a feasible solution and $\sigma \in\mathfrak{S}_n$. 
\begin{itemize}
    \item[\eqref{eq:delsarte_lin:C1}] $f(\sigma \cdot \bm{0}) = f(\bm{0}) = 1$.
    \item[\eqref{eq:delsarte_lin:C2}] By Proposition \ref{prop:partial_fourer_under_Sn} from Section \ref{section:partial_fourier_transform_symmetries} below,
    if $\mathcal{F}_S(f) \geq 0$ then also $\mathcal{F}_S(\sigma \circ f)\geq 0$.
    \item[\eqref{eq:delsarte_lin:C3}] Row weights are invariant under column permutations.
    \item[\eqref{eq:delsarte_lin:C4}] Permuting of the columns of $X$ is equivalent to multiplication from the right
    by a permutation matrix $P$.
    Since matrix multiplication is associative, 
    \[
        (\sigma \circ f)(TX)=f(T(XP))=f(XP)=(\sigma\circ f)(X)
    \]
    for every $T\in\GL{r}{2}$.
    \item[\eqref{eq:delsarte_lin:objective}]  (also \eqref{eq:delsarte_lin:objective_r}) Permutation only affects the order of summation, but not the total sum.
\end{itemize}
Hence, $\sigma \circ f$ is a feasible solution with the same value as $f$.

Therefore, there is no loss in restricting to symmetric solutions of $\DelsarteLinCubeNoParams$,
i.e., to solutions $f$ that are constant on the orbits $\cube{r\times n}/\mathfrak{S}_n$.
Such solutions can be expressed as a linear combination of orbit indicators:
\[
    f(X) = \sum_{Orb\in\cube{r\times n}/\mathfrak{S}_n} \varphi_{Orb} \cdot \1_{Orb}(X)
\]
where $\1_{Orb}:\cube{r\times n}\to\{0,1\}$ is the indicator function of the set 
$Orb\in\cube{r\times n}/\mathfrak{S}_n$, and
$(\varphi_{Orb})$ are real numbers.
To exploit this symmetry we 
reformulate the LP in terms of $(\varphi_{Orb})$.

The following definition will be useful in depicting the set of orbits.
\begin{definition}\label{def:column_frequency} Let $\bm{\xi}_1,\dots,\bm{\xi}_n\in\cube{r}$
be the columns of $X \in \cube{r\times n}$. The
    \textbf{column enumerator} of $X$ counts how many times
    each vector in $\cube{r}$ appears as a column in $X$:
    \[
        \cf_X\in \N^{2^r},\quad
        \cf_X(\bm{u}) = |\{ 1\leq i \leq n: \bm{\xi}_i = \bm{u}\}|
    \]
\end{definition}
Observe that when $r=1$, $\cf_{\bm{x}}(1) = |\bm{x}|$ and $\cf_{\bm{x}}(0) =
n-|\bm{x}|$.

The column enumerator of a matrix clearly determines its orbit, i.e.,
$\mathfrak{S}_n \cdot \x = \mathfrak{S}_n\cdot \y \iff
\cf_X=\cf_Y$. The set of orbits $\cube{r\times n}/\mathfrak{S}_n$ is therefore
isomorphic to the set of all possible column
enumerators, which we denote by
$\mathcal{I}_{r,n}$, 
\begin{equation}
    \label{eq:symmetrized_index_set}
    \mathcal{I}_{r,n} \coloneqq
    \{\bm{\alpha}=(\alpha_0,\dots,\alpha_{2^r-1}):
    \alpha_i \in \N,~
    \sum_{i=0}^{2^r-1}\alpha_i = n \}
\end{equation}
Equivalently, it is the set of ordered partitions of $n$ into $2^r$ parts. 
In the sequel, we will introduce a different equivalent way of looking at $\mathcal{I}_{r,n}$.

The level-set indicator function of ${\bm{\alpha}} \in \mathcal{I}_{r,n}$ is defined via
\[
    L_{\bm{\alpha}}:\cube{r\times n}\to \{0,1\},\quad
    L_{\bm{\alpha}}(X) = 
    \begin{cases}
        1 & \cf_X = {\bm{\alpha}} \\
        0 & \text{o/w}
    \end{cases}
\]
This allows us to express any symmetric solution to $\DelsarteLinCube{r}{n}{d}$ as follows:
\[
    f = \sum_{\bm{\alpha}\in\mathcal{I}_{r,n}} \varphi_{\bm{\alpha}} L_{\bm{\alpha}}
\]
We need to introduce some more notation. Let
$\epsilon_{\bm{u}}:\cube{r}\to\R$ be the indicator of $\bm{u}$. Namely,
$\epsilon_{\bm{u}}(\bm{v}) = 1$ if $\bm{v}=\bm{u}$ and $0$ otherwise, for
$\bm{v}\in\cube{r}$. Note the distinction between indicators
in $\R^{\cube{r}}$, and those in
$\cube{r}$, which we denote by $\bm{e}_i$, for $i=1,\dots,r$. We write, for
example,
\[
    \bm{\alpha} = 
    (n-k)\epsilon_{\bm{0}} + k\epsilon_{\bm{e}_i}
    \in \mathcal{I}_{r,n}
\]
Here, $\bm{0},\bm{e}_i\in\cube{r}$, and $k$ is an integer between $0$ and $n$.

Every $\bm{\alpha}\in\mathcal{I}_{r,n}$ is also considered as a real function
on $\cube{r}$. Namely, $\alpha_{\bm{u}}$ is synonymous with $\alpha_i$, where
$\bm{u}\in\cube{r}$ is the binary representation of $i\in\N$.
As a Boolean function, we apply Fourier transform to $\bm{\alpha}$:
$\hat{\bm{\alpha}}_{\bm{u}} = \langle \chi_{\bm{u}},\bm{\alpha}\rangle$,
for every $\bm{u}\in\cube{r}$.

Let us now rewrite $\DelsarteLinCube{r}{n}{d}$ in terms of 
$(\varphi_{\bm{\alpha}})_{\bm{\alpha}\in\mathcal{I}_{r,n}}$.
\begin{itemize}
    \item[\eqref{eq:delsarte_lin:C1}] The orbit of $\bm{0}\in\cube{r\times n}$ contains only the element $\bm{0}$,
    so $f(\bm{0})=1$ implies $\varphi_{n\epsilon_{\bm{0}}}=1$.
    \item[\eqref{eq:delsarte_lin:C2}] By linearity of (partial) Fourier transform,
    \[
        \mathcal{F}_S(f)(X) = 
        \sum_{\bm{\alpha}\in\mathcal{I}_{r,n}} \varphi_{\bm{\alpha}} \cdot \mathcal{F}_S(L_{\bm{\alpha}})(X)
    \]
    for every $S\subset[r]$ and $X\in\cube{r\times n}$.

    In Section \ref{section:partial_fourier_transform_symmetries} below, we show that 
    $\mathcal{F}_S(L_{\bm{\alpha}})(X)$ depends only on the column enumerator of $X$.

    When $S=[r]$, namely for $\hat{L}_{\bm{\alpha}}(X)$, it turns out that it is a
    multivariate polynomial in $\Gamma_X$. 
    In Section \ref{section:Krawtchouk_properties} we denote $\hat{L}_{\bm{\alpha}}(X) 
    \coloneqq K_{\bm{\alpha}}(\Gamma_X)$, and show that $\{K_{\bm{\alpha}}\}_{\bm{\alpha}\in\mathcal{I}_{r,n}}$
    is a set of polynomials over $\R^{2^r}$ orthogonal w.r.t.\ the multinomial distribution. These
    polynomials are called {\em multivariate Krawtchouks}.

    For $S\neq [r]$, we denote $\mathcal{F}_S(L_{\bm{\alpha}})(X) \coloneqq K^S_{\bm{\alpha}}(\Gamma_X)$.
    We call the set $\{K^S_{\bm{\alpha}}\}_{\bm{\alpha}\in\mathcal{I}_{r,n}}$ {\em partial Krawtchouks}.
    These are orthogonal functions w.r.t.\ an appropriate measure, though not polynomials. 
    In Section \ref{section:Krawtchouk_properties}
    we describe these functions as products of multivariate Krawtchouks.

    Constraint \eqref{eq:delsarte_lin:C2} implies
    \[
        \sum_{\bm{\alpha}\in\mathcal{I}_{r,n}} \varphi_{\bm{\alpha}} K^S_{\bm{\alpha}} \geq 0
    \]
    for every $S\subset[r]$.

    \item[\eqref{eq:delsarte_lin:C3}] The following proposition expresses the weights of the row space of $X$ in terms of its column enumerator.
    \begin{proposition}\label{prop:weight_equivalence} For every
        $X=(x_{i,j})\in\cube{r\times n}$ and $\bm{u}\in\cube{r}$,
        \[
            |\bm{u}^\intercal X| =\frac{1}{2}\big( n - 2^{r} \widehat{\cf}_X(\bm{u})\big)
        \]
        where $\widehat{\cf}_X(\bm{u})$ is the Fourier transform of $\cf_X$ at
        $\bm{u}$.
    \end{proposition}
    \begin{proof}
        By definition, $\bm{u}^\intercal X \in \cube{n}$ and $|\bm{u}^\intercal X|
        = \sum_{j=1}^{n}(\bm{u}^\intercal X)_j$, where $(\bm{u}^\intercal X)_j$ is
        the $j$-th bit and the sum is over the integers. Concretely, for
        $j=1,\dots,n$:
        \begin{dmath*}
            (\bm{u}^\intercal X)_j
            = \sum_{i:u_i=1}x_{i,j}\bmod 2
            = \frac{1}{2}\Big(1-(-1)^{\sum_{i: u_i=1} x_{i,j}}\Big)
            = \frac{1}{2}\Big(1-(-1)^{\langle \bm{u}, \bm{\xi}_j \rangle}\Big)
        \end{dmath*}
        where $\bm{\xi}_j = (x_{i,j})_{i=1}^{r}$ is the $j$-th column of $X$. Thus
        \[
            |\bm{u}^\intercal X|
            =\frac{1}{2} \sum_{j=1}^{n} \Big(1-(-1)^{\langle \bm{u}, \bm{\xi}_j \rangle}\Big)
        \]
        But $\bm{\xi}_j$ appears $\cf_X(\bm{\xi}_j)$ times in $X$, so grouping the
        summands by column, we have
        \[
            |\bm{u}^\intercal X|
            =\frac{1}{2} \sum_{\bm{v}\in\cube{r}} \cf_X(\bm{v}) (1-(-1)^{\langle \bm{u},\bm{v}\rangle})
            =\frac{1}{2}\big(n - \chi_{\bm{u}}^\intercal \cf_X\big)
        \]
    \end{proof}
    Thus, we require that $\varphi_{\bm{\alpha}}=0$ whenever
    $1\leq \frac{1}{2}\big( n - 2^{r} \widehat{\bm{\alpha}}_{\bm{u}}\big)\leq d-1$
    for some $\bm{u}\in\cube{r}$.

    \item[\eqref{eq:delsarte_lin:C4}] When $X\in\cube{r\times n}$ gets multiplied 
    on the left by $T\in \GL{r}{2}$, its column enumerator, $\cf_X$ gets modified.
    Here we need to define 
    the action of $T$ on $\mathcal{I}_{r,n}$, in a way that is consistent with
    this modification. Indeed, define
    \[
        (T\cdot \bm{\alpha})_{\bm{u}} = \alpha_{T^{-1}\bm{u}}
    \]
    This ensures $T\cdot \Gamma_X = \Gamma_{TX}$.
    
    \item[\eqref{eq:delsarte_lin:objective}] The vector $(\bm{x},0,\dots,0)\in\cube{rn}$ corresponds
    to the matrix $\bm{e}_1 \bm{x}^\intercal\in\cube{r\times n}$. Say $|\bm{x}| = k$. Then,
    its column enumerator is $\cf_{\bm{e}_1 \bm{x}^\intercal} = (n-k)\epsilon_{\bm{0}} + k \epsilon_{\bm{e}_1}$.
    The orbit of $\bm{e}_1 \bm{x}^\intercal$ has cardinality $\binom{n}{k}$. 
    Hence, the objective function becomes
    \[
        \textup{maximize} \quad \sum_{k=0}^{n}\binom{n}{k} \varphi_{(n-k)\epsilon_{\bm{0}} + k \epsilon_{\bm{e}_1}}
    \]
    
    \item[\eqref{eq:delsarte_lin:objective_r}] Summing over the entire set $\mathcal{I}_{r,n}$ with multiplicites,
    \[
        \textup{maximize} \quad \sum_{\bm{\alpha}\in\mathcal{I}_{r,n}} \binom{n}{\bm{\alpha}} \varphi_{\bm{\alpha}}
    \]
    where $\binom{n}{\bm{\alpha}}$ is the multinomial coefficient.
\end{itemize}
Let us now define the symmetrized version of $\DelsarteLinCubeNoParams$.

\begin{definition}\label{def:delsarte_lin} $\DelsarteLin{r}{n}{d}$:

\begin{alignat*}{3}
    & \mathrlap{ \underset{{\varphi:\mathcal{I}_{r,n}\to\R}}
        {\textup{maximize}}\quad
    \sum_{k=0}^{n} 
        \binom{n}{k}
        \varphi_{(n-k)\epsilon_{\bm{0}} + k\epsilon_{\bm{e}_1}}}
    \tag{$Obj_{/\mathfrak{S}_n}$}
    \\
    &\st \\
    &\quad \varphi_{n\epsilon_{\bm{0}}} = 1 \tag{$C1_{/\mathfrak{S}_n}$}
    \\
    &\quad \sum_{{\bm{\alpha}}\in\mathcal{I}_{r,n}} \varphi_{\bm{\alpha}}
        K^{S}_{{\bm{\alpha}}}(\bm{\beta}) \geq 0 \quad 
        &&
        \forall S\subset [r],~
        \bm{\beta}\in \mathcal{I}_{r,n} \tag{$C2_{/\mathfrak{S}_n}$}
        \label{eq:delsarte_lin_sym:C2}
    \\
    &\quad \varphi_{\bm{\alpha}} = 0  
        &&\text{if } 
        1\leq \frac{1}{2}(n-2^r \hat{\bm{\alpha}}_{\bm{e}_1})\leq d-1
        \tag{$C3_{/\mathfrak{S}_n}$} \label{eq:delsarte_lin_sym:C3} \\
    &\quad \varphi_{\bm{\alpha}} = \varphi_{T\cdot {\bm{\alpha}}} 
        && \forall T
        \in \GL{r}{2}
        \tag{$C4_{/\mathfrak{S}_n}$} 
\end{alignat*}
\end{definition}
We also mention 
two important variations, \eqref{eq:delsarte_lin:C2_weak} and
\eqref{eq:delsarte_lin:objective_r}:
\begin{alignat*}{2}
    & \maximize\quad
    \sum_{\bm{\alpha}\in\mathcal{I}_{r,n}} \varphi_{\bm{\alpha}}
    \tag{$Obj'_{/\mathfrak{S}_n}$}
    \\
    & \varphi \geq 0;\quad
    \sum_{{\bm{\alpha}}\in\mathcal{I}_{r,n}} \varphi_{\bm{\alpha}}
    K_{{\bm{\alpha}}}(\bm{\beta}) \geq 0 \quad 
    & \forall 
    \bm{\beta}\in \mathcal{I}_{r,n} \tag{$C2'_{/\mathfrak{S}_n}$}
\end{alignat*}

By the comments from the beginning of this section, we have the following equivalence.
\begin{proposition}\label{prop:delsarte_lin_Sn_equivalence} For every positive
    integers $r,n,d$, such that $d\leq n/2$,
    \[
        \val \DelsarteLin{r}{n}{d}
        = \val \DelsarteLinCube{r}{n}{d}
    \]
\end{proposition}

Note that $\DelsarteLin{1}{n}{d}$ is identical to Delsarte's LP. Observe
that $\mathcal{I}_{1,n}$ is isomorphic to the set $\{0,1,\dots,n\}$. Rewrite
the LP with a new set of variables, $a_k \coloneqq
\binom{n}{k}\varphi_{(n-k)\epsilon_0 + k\epsilon_1}$, for $k=0,1,\dots,n$.
Using the Krawtchouk symmetry identity,
$\binom{n}{j}K_i(j)=\binom{n}{i}K_j(i)$, transform the Krawtchouk constraint
\eqref{eq:delsarte_lin_sym:C2} as follows:
\[
    \sum_{j=0}^{n}\binom{n}{j}^{-1}a_j K_{j}(i)
    = \binom{n}{i}^{-1} \sum_{j=0}^{n}a_j K_{i}(j)
\]
The result is Delsarte's LP:
\begin{definition}\label{def:delsarte_lp} $\Delsarte{n}{d}$:
    \begin{alignat*}{2}
        &\mathrlap{\underset{a_0,\dots,a_n\in\R}
        {\textup{maximize}}\quad \sum_{i=0}^{n} a_i}
            \tag{$obj_{/\mathfrak{S}_n}$}\label{eq:delsarte_lp:objective} \\
        &\st \\
        &\quad a_0 = 1
        \tag{$d1_{/\mathfrak{S}_n}$}\label{eq:delsarte_lp:c1} \\
        &\quad a_i \geq 0; \quad
            \sum_{j=0}^{n}a_j K_{i}(j) \geq 0, \quad 0\leq i \leq n
        \tag{$d2_{/\mathfrak{S}_n}$}\label{eq:delsarte_lp:c2} \\
        &\quad a_i = 0 
            \quad \text{\textup{if} } 1\leq i\leq d-1
            \tag{$d3_{/\mathfrak{S}_n}$}\label{eq:delsarte_lp:c3}
    \end{alignat*}
\end{definition}

\section{On Partial Fourier Transform}

\label{section:partial_fourier_transform_symmetries}

In this section we explore interactions between the groups $\mathfrak{S}_n$ and
$\GL{r}{2}$ and the partial Fourier transform. The former, $\mathfrak{S}_n$
acts on $\cube{r\times n}$ by permuting columns. The latter, $\GL{r}{2}$
acts on $\cube{r\times n}$ by matrix multiplication from the left. The group of
order-$r$ permutation matrices is a subgroup of $\GL{r}{2}$ which acts on
$\cube{r\times n}$ by permuting rows.

We recall our dual view of $\cube{rn}$, once as a concatenation of $r$ vectors 
$\bm{x}_1,\dots,\bm{x}_r\in\cube{n}$,
and once as a matrix $X\in \cube{r\times n}$ whose rows are the above vectors.
If the group
$G$ acts on $\cube{r\times n}$, and $g\in G$, we denote $(f\circ g)(\x) =
f(g\cdot \x)$ for any $\x\in \cube{r\times n}$ and $f:\cube{rn}\to\R$.

The proofs for some of the following propositions appear in the appendix.

\begin{proposition} \label{prop:partial_fourer_under_Sn}
    Let $\sigma\in \mathfrak{S}_n$, $\x\in\cube{r\times n}$, $S\subset[r]$, and
    $f:\cube{r\times n}\to\R$. Then,
    \[
        \mathcal{F}_S(f\circ \sigma)
        = \mathcal{F}_S(f) \circ \sigma
    \]
\end{proposition}

\begin{proposition} \label{prop:partial_fourer_under_Sr}
    Let $\pi\in \mathfrak{S}_r$ act on the set $\cube{r\times n}$ by row
    permutation. Let $\x\in\cube{r\times n}$, $S\subset[r]$, and
    $f:\cube{r\times n}\to\R$. Then,
    \[
        \mathcal{F}_S(f\circ \pi)
        = \mathcal{F}_{\pi^{-1}(S)}(f) \circ \pi
    \]
\end{proposition}

\begin{proposition} \label{prop:partial_fourer_under_GLr}
    Let $T\in \GL{r}{2}$ be the elementary matrix of row addition, mapping
    $\bm{e}_i\mapsto \bm{e}_i+\bm{e}_j$, for some $i,j\in [r]$, $i\neq j$,
    and $\bm{e}_k\mapsto
    \bm{e}_k$ for $k\neq i$, where $\bm{e}_k\in\cube{r}$ is the $k$-th standard basis
    vector. Let $X\in\cube{r\times n}$, $S\subset[r]$, and $f:\cube{r\times
    n}\to\R$. Then,
    \begin{itemize}
        \item if $i,j\in S$:
        \[
            \mathcal{F}_S(f\circ T)
            = \mathcal{F}_S(f)\circ T^{\intercal}
        \]
        \item if $i,j\notin S$:
        \[
            \mathcal{F}_S(f\circ T)
            = \mathcal{F}_S(f) \circ T
        \]
        \item if $i\in S, j\notin S$:
        \[
            \mathcal{F}_S(f\circ T)(X)
            = \chi_{\bm{x}_i}(\bm{x}_j) \mathcal{F}_S(f)(X)
        \]
    \end{itemize}
\end{proposition}
Note that do not consider the case $i\notin S, j\in S$, since the expression
does not simplify in that case.

\begin{lemma}\label{lemma:C456_equivalence} Let $f:\cube{r\times n}\to \R$. The
    following are equivalent:
    \begin{enumerate}
        \item \label{C456_equivalence:item1} For every $T\in\GL{r}{2}$,
        \[f = f\circ T.\]
        \item \label{C456_equivalence:item2} For every $S\subset[r]$, 
        \[ \mathcal{F}_S(f)(\bm{x}_1,\dots,\bm{x}_r) = 0,\]
        if 
        $\langle \bm{x}_i, \bm{x}_j \rangle = 1 \bmod 2$ for some
        $i\in S$ and $j\in [r]\setminus S$.
        \item \label{C456_equivalence:item3} For every $S\subset[r]$, 
        \[
            \mathcal{F}_{S}(f) =
            \mathcal{F}_{S}(f) \circ (T_1 T_2)\]
        if $T_1,T_2 \in \GL{r}{2}$, and
        $T_1 \bm{e}_i = \bm{e}_i$ for every $i\in S$, $T_2 \bm{e}_i = \bm{e}_i$
        for every $i\in [r]\setminus S$.
    \end{enumerate}
\end{lemma}
\begin{proof}
    \begin{itemize}
        \item $\eqref{C456_equivalence:item1} \Rightarrow
        \eqref{C456_equivalence:item2}$: Let $S\subsetneq [r]$, $S\neq
        \emptyset$. Let $\bm{x}_1,\dots,\bm{x}_r \in \cube{n}$ and $i\in S$, $j\in
        [r]\setminus S$ s.t. $\langle \bm{x}_i,\bm{x}_j \rangle = 1 \bmod 2$. Let $T\in
        \GL{r}{2}$ be the mapping $\bm{x}_i\mapsto \bm{x}_i + \bm{x}_j$ and
        $\bm{x}_k\mapsto \bm{x}_k$ for $k\neq i$.
        By assumption and by proposition \ref{prop:partial_fourer_under_GLr},
        \[
            \mathcal{F}_{S}(f) = \mathcal{F}_{S}(f\circ T)
            = \chi_{\bm{x}_i}(\bm{x}_j) \mathcal{F}_{S}(f)
            = -\mathcal{F}_{S}(f)
        \]
        Hence $\mathcal{F}_{S}(f)=0$.
        \item $\eqref{C456_equivalence:item2} \Rightarrow
        \eqref{C456_equivalence:item3}$: It is enough to show that
        $\mathcal{F}_S(f)$ is invariant under the mapping $\bm{x}_i \mapsto
        \bm{x}_i+\bm{x}_j$, where $i\neq j$ and $i,j$ are either both in $S$ or both in
        $[r]\setminus S$. The rest follows by composition of such operators.
        
        If $|S|\leq 1$ the claim holds trivially. Otherwise, let $i\neq j$, $i,j\in S$,
        and let $\bm{x}_1,\dots,\bm{x}_r\in \cube{n}$. Observe
        that $\mathcal{F}_{S}(f) = \mathcal{F}_{\{i\}}\mathcal{F}_{S\setminus
        \{i\}}(f)$. Hence
        \begin{dmath*}
            \mathcal{F}_{S}(f)(\bm{x}_1,\dots,\bm{x}_r)
            = 2^{-n}\sum_{\bm{y}\in\cube{r}} \chi_{\bm{x}_i}(\bm{y}) 
                \mathcal{F}_{S\setminus \{i\}}(f)
                    (\dots,\bm{x}_{i-1},\bm{y},\bm{x}_{i+1},\dots)
        \end{dmath*}
        by assumption, $\mathcal{F}_{S\setminus
        \{i\}}(f)(\dots,\bm{x}_{i-1},\bm{y},\bm{x}_{i+1},\dots) = 0$ if
        $\langle \bm{y},\bm{x}_j\rangle = 1$, hence $\chi_{\bm{x}_j}(\bm{y})=1$ 
        for every non-zero element of the sum. So
        \begin{dmath*}
            \mathcal{F}_{S}(f)(\bm{x}_1,\dots,\bm{x}_r)
            = 2^{-n}\sum_{\bm{y}\in\cube{r}}
                \chi_{\bm{x}_i}(\bm{y}) \chi_{\bm{x}_j}(\bm{y})
                \times
                \\
                \times
                \mathcal{F}_{S\setminus \{i\}}(f)
                    (\dots,\bm{x}_{i-1},\bm{y},\bm{x}_{i+1},\dots)
            = \mathcal{F}_{S}(f)(\dots,\bm{x}_{i-1},\bm{x}_i+\bm{x}_j,\bm{x}_{i+1},\dots)
        \end{dmath*}
        To see that the same applies if $i,j\in[r]\setminus S$, observe that
        $\mathcal{F}_{S}(f) = 2^n\mathcal{F}_{\{i\}}\mathcal{F}_{S\cup
        \{i\}}(f)$ and repeat the same steps.
        
        \item $\eqref{C456_equivalence:item3} \Rightarrow
        \eqref{C456_equivalence:item1}$: Take $S = \emptyset$.
    \end{itemize}
\end{proof}

\section{On Multivariate Krawtchouk Polynomials}

\label{section:Krawtchouk_properties}

The multivariate Krawtchouk polynomials are orthogonal polynomials on the
multinomial distribution. Univariate Krawtchouk polynomials
are the Fourier transform of the level sets in the Boolean cube, and as we show
in this section, these polynomials are the Fourier transform of the level-set indicators $\{L_\alpha\}$.

We borrow the terminology of \cite{diaconis2014introduction}.
The multinomial distribution $m(\bm{\alpha},\bm{p})$
arises in the stochastic process where $n$ identical balls
are independently dropped into $d$ bins, where the probability
of falling into the $i$-th bin is $p_i$.
The probability that $\alpha_i$ balls end up in bin $i$ is 
\[
    m(\bm{\alpha},\bm{p})
    = \binom{n}{\alpha_0,\dots,\alpha_{d-1}}\prod_{i=0}^{d-1}p_i^{\alpha_i}
    = \binom{n}{\bm{\alpha}}\bm{p}^{\bm{\alpha}}
\]
Here $\bm{p}=(p_0,\dots,p_{d-1})$, all $\alpha_i$ are nonnegative integers and their sum is $n$.
We use the shorthand $m(\bm{\alpha})$ when $\bm{p}$ is uniform.

Orthogonal systems of univariate polynomials are constructed
by applying a Gram-Schmidt process to
the polynomials $1, x, x^2,\ldots$ e.g., \cite{szeg1939orthogonal}.
The result depends only on a measure that we fix on the underlying set.
However, as mentioned e.g., in \cite{dunkl2014orthogonal},
in the process of defining an orthogonal multivariate 
family of polynomials, there is another choice to make, and this choice affects
the resulting family. Namely, we need to choose the
order in which we go over the monomials of a given degree.
In \cite{diaconis2014introduction}, this freedom is mitigated by
choosing a basis of orthogonal
functions on $\{0,1,\dots,d-1\}$. Every such basis leads to 
a unique set of orthogonal polynomials, as follows. Let
$\bm{h}=\{h^{l}\}_{l=0}^{d-1}$ be a complete set of orthogonal functions 
w.r.t.\ $\bm{p}$, with $h^0\equiv 1$. Namely,
\[
    \sum_{i=0}^{d-1}h^{l}(i)h^{k}(i)p_i = \delta_{lk} a_k,
    \quad
    0\leq k,l\leq d
\]
The Krawtchouks are defined in terms of a generating function.
Fix $\bm{\alpha}$ and $\bm{h}$ as above. For every choice of nonnegative reals
$\xi_0,\ldots,\xi_{d-1}$ whose sum is $n$, we define
\[
    Q(\bm{\xi})=Q_{ \bm{\alpha}}(\bm{\xi},\bm{h})
    = \underset{\prod_{i=1}^{d-1}w_i^{\alpha_i}}{\text{coef}}
    \prod_{j=0}^{d-1}\bigg\{
        1 + \sum_{l=1}^{d-1}w_l h^{l}(j)
    \bigg\}^{\xi_j}
\]
where $\bm{w}=(w_0\dots,w_{d-1})$ are formal variables.
The total degree of
$Q_{ \bm{\alpha}}$ is $\sum_{i=1}^{d-1}\alpha_i$. Note that
$\alpha_0,w_0$ do not appear in the definition. An equivalent definition that
does include $\alpha_0,w_0$ is:
\begin{equation}
    \label{eq:Krawtchouk_def}
    Q_{ \bm{\alpha}}(\bm{\xi},\bm{h})
    = \underset{\bm{w}^{\bm{\alpha}}}{\text{coef}}
    \prod_{j=0}^{d-1}\bigg\{
        \sum_{l=0}^{d-1}w_l h^{l}(j)
    \bigg\}^{\xi_j}
\end{equation}
It is easy to see the equivalence by expanding each factor with the multinomial
expansion. We will be using all of this with $d=2^r$, uniform $\bm{p}\equiv 2^{-r}$
and with
the orthonormal functions that are the characters of $\cube{r}$:
$\bm{h} = \{\chi_{\bm{u}}\}_{{\bm{u}}\in\cube{r}}$.

Recall the definition of level-set indicators, $\{L_{\bm{\alpha}}\}$:
\[
    L_{\bm{\alpha}}(\x) = \1_{[\cf_\x = \bm{\alpha}]},\quad
    \x\in \cube{r\times n}
\]
We also defined $\mathcal{I}_{r,n}$ the set of all ordered partitions of $[n]$ into $2^r$ parts.
\[
    \mathcal{I}_{r,n} = \{\cf_\x: \x\in\cube{r\times n}\}
\]
Also, $\mathcal{I}_{r,n}$ is the support of the multinomial distribution
with $n$ balls, $2^r$ bins, where $\bm{p}$ is uniform.

Let $\x\in \cube{r\times n}$ be a random matrix that results by sampling $n$ columns
independently and uniformly
from $\cube{r}$. The probability that $L_{\bm{\alpha}}(\x)=1$ is 
$m(\bm{\alpha}) = 2^{-rn}\binom{n}{\alpha}$. It is clear that
$L_{\bm{\alpha}}(\x)$ depends only on $\cf_\x$, and by proposition
\ref{prop:partial_fourer_under_Sn} this is true for $\hat{L}_{\bm{\alpha}}(\x)$
as well.
Define
\[
    K_{\bm{\alpha}}(\cf_\x) = 2^{rn} \hat{L}_{\bm{\alpha}}(\x),\quad
    \x\in\cube{r\times n}
\]
It is easy to see that $\{K_{\bm \alpha}\}$ are orthogonal with respect to $m(\bm{\alpha})$, using Parseval's identity:
\begin{dmath*}
    \sum_{\bm{\gamma}} m(\bm{\gamma}) K_{\bm{\alpha}}(\bm{\gamma})
        K_{\bm{\beta}}(\bm{\gamma}) 
    = \sum_{\bm{\gamma}} 2^{-rn}\sum_{\x:\cf_X=\bm{\gamma}}2^{2rn}
        \hat{L}_{\bm{\alpha}}(\x) \hat{L}_{\bm{\beta}}(\x)
    = \sum_{\x\in\cube{r\times n}}
        L_{\bm{\alpha}}(\x) L_{\bm{\beta}}(\x)
    =  \binom{n}{\bm\alpha}\delta_{\bm{\alpha},\bm{\beta}}
\end{dmath*}
where $\bm{\alpha},\bm{\beta},\bm{\gamma}\in\mathcal{I}_{r,n}$.
The extra $2^{-rn}$ is there because inner product is normalized in the
non-Fourier space.

The following proposition shows that $\{K_{\bm\alpha}\}$ are Krawtchouk
polynomials.

\begin{proposition}
    $K_{\bm \alpha}$ is the Krawtchouk polynomial $Q_{\bm \alpha}(\cdot,\bm{h})$ with
    $d=2^r$, $\bm{h}$ are the Fourier characters $\{\chi_{\bm{u}}\}_{{\bm{u}}\in\cube{r}}$,
    and $\bm{p}\equiv 2^{-r}$ is the uniform distribution.
    \label{prop:level_set_fourier_Krawtchouk}
\end{proposition}
\begin{proof}
    Let $\x\in\cube{r\times n}$ and $\cf_\x=\bm{\beta}=(\beta_{\bm{u}})_{{\bm{u}}\in\cube{r}}$.
    We show that $K_{\bm \alpha}(\bm{\beta})$ coincides
    with the definition of $Q_{\bm \alpha}(\bm{\beta},\{\chi_{\bm{u}}\}_{{\bm{u}}\in\cube{r}})$ in \eqref{eq:Krawtchouk_def}.
    
    By definition,
    \begin{align*}
        K_{\bm \alpha}(\bm{\beta}) = \hat{L}_{\bm \alpha}(\x)
        &= 2^{-rn} \sum_{\y\in\cube{r\times n}} (-1)^{\langle X,Y \rangle} L_{\bm\alpha}(\y)
    \end{align*}
    The inner product between $\x$ and $\y$ can be expressed column-wise,
    \[
        \langle X,Y \rangle
        = \sum_{j=1}^{n}\langle (\x^\intercal)_j,(\y^\intercal)_j \rangle
        = \sum_{{\bm{u}},{\bm{v}}\in\cube{r}} \langle{\bm{u}},{\bm{v}} \rangle \cf_{[X,Y]}({\bm{u}},{\bm{v}})
    \]
    where $[X,Y]\in\cube{2r\times n}$ is the stacking of $\x$ on top of $\y$,
    and $\cf_{[X,Y]}({\bm{u}},{\bm{v}})$ is the number of times the column $[{\bm{u}},{\bm{v}}]\in\cube{2r}$ 
    appears in the matrix $[X,Y]$. We consider 
    $\cf_{[X,Y]}({\bm{u}},{\bm{v}})$ as a matrix indexed by $\cube{r}\times \cube{r}$.
    Its ${\bm{u}}$-th row sums to $\beta_{\bm{u}}$ and its ${\bm{v}}$-th column sums
    to $\cf_Y({\bm{v}})$. Hence, 
    \begin{gather*}
        \hat{L}_{\bm \alpha}(\x)=
        2^{-rn} \sum_{A} \sum_{\substack{\y\in\cube{r\times n}: \\ \cf_{[X,Y]}=A} }
            \prod_{{\bm{u}},{\bm{v}}\in\cube{r}} (-1)^{\langle {\bm{u}},{\bm{v}} \rangle A_{{\bm{u}},{\bm{v}}}} L_{\bm{\alpha}}(\y)
    \end{gather*}
    where the outer sum is over all matrices $A\in \N^{2^r\times 2^r}$ with
    $A\cdot \1=\bm{\beta}$.
    
    If $\cf_\y=\bm{\alpha}$ then $\1^{\intercal}\cdot A = \bm{\alpha}$.
    In particular, $A$ uniquely determines $L_{\bm\alpha}(\y)$,
    so
    the product does not depend on $\y$. The sum over $\y$ evaluates to the size
    of the set $\{\y\in\cube{r\times n} : \cf_{[X,Y]}=A\}$, which we now compute.
    For every ${\bm{u}}\in\cube{r}$, $[X,Y]$ contains
    $\cf_{X}({\bm{u}})=\beta_{\bm{u}}$ columns whose prefix is ${\bm{u}}$. For every ${\bm{v}}\in\cube{r}$, the column
    $[{\bm{u}},{\bm{v}}]$ appears $A_{{\bm{u}},{\bm{v}}}$ times in $[X,Y]$. Since the position of the ${\bm{u}}$'s
    is fixed, it is left to position the ${\bm{v}}$'s with respect to each $\bm{u}$. Hence
    \[
        |\{\y\in\cube{r\times n} : \cf_{[X,Y]}=A\}|
        = \prod_{{\bm{u}}\in\cube{r}} \binom{\beta_{\bm{u}}}{A_{\bm{u}}}
    \]
    where $A_{\bm{u}}$ is the row of $A$ that is indexed by $\bm{u}$.
    
    Let $\bm{w}=(w_{\bm{v}})_{{\bm{v}}\in\cube{r}}$ be formal variables.
    If $A$ is such that $\1^\intercal A = \bm{\alpha}$ then
    \[
        \prod_{{\bm{u}},{\bm{v}}}w_{\bm{v}}^{A_{{\bm{u}},{\bm{v}}}}
        = \prod_{\bm{v}}w_{\bm{v}}^{\sum_{\bm{u}} A_{{\bm{u}},{\bm{v}}}}
        = \prod_{\bm{v}}w_{\bm{v}}^{\alpha_{\bm{v}}}
    \]
    hence $\hat{L}_{\bm \alpha}(X)$ equals
    \[
        \underset{\bm{w}^{\bm{\alpha}}}{\text{coef}}~
            2^{-rn} \sum_{A} 
            \prod_{{\bm{u}}\in\cube{r}} \binom{\beta_{\bm{u}}}{A_{\bm{u}}} \prod_{{\bm{v}}\in\cube{r}}
            \Big(
                (-1)^{\langle {\bm{u}},{\bm{v}} \rangle} w_{\bm{v}}
            \Big)^{A_{{\bm{u}},{\bm{v}}}} 
    \]
    The sum over $A$ can be expanded to nested sums over its rows,
    \[
        \sum_{A} = \sum_{A_0}\sum_{A_1}\cdots\sum_{A_{2^{r}-1}}
    \]
    where $A_{\bm{u}}\in\N^{2^r}$, $A_{\bm{u}}\cdot\1 = \beta_{\bm{u}}$. Every factor in the
    product depends on a single row of $A$,
    so the product and the sum can be transposed.
    Then, by the multinomial theorem:
    \begin{dmath*}
        \hat{L}_{\bm \alpha}(X)
        = \underset{\bm{w}^{\bm{\alpha}}}{\text{coef}}~
            2^{-rn}
            \prod_{{\bm{u}}\in\cube{r}}
            \bigg\{
                \sum_{{\bm{v}}\in\cube{r}} (-1)^{\langle {\bm{u}},{\bm{v}} \rangle} w_v
            \bigg\}^{\beta_{\bm{u}}}
        = 2^{-rn} Q_{\bm{\alpha}}(\bm{\beta}, \{\chi_{\bm{v}}\}_{\bm{v}\in\cube{r}})
    \end{dmath*}
\end{proof}

We turn to deal with the partial Fourier transform of the level-set indicators.
Let $S\subsetneq\{1,\dots,r\}$ be non-empty. 
Denote by $X',X''$ be the sub-matrix of $X\in \cube{r\times n}$ with
row set $S$ and $[r]\setminus S$, respectively.
Similarly, $\bm{u}',\bm{u}''$ are obtained from $\bm{u}\in\cube{r}$ by restricting to
$S,[r]\setminus S$, respectively.
For $\bm{\alpha} \in \mathcal{I}_{r,n}$ define the rearrangement of 
$\bm{\alpha}$ into a matrix
$\bm{\alpha}^{S} = (\alpha^{S}_{\bm{u}'',\bm{u}'})_{\bm{u}''\in\cube{r-|S|},\bm{u}'\in\cube{|S|}}$ by
\[
    \alpha^{S}_{\bm{u}'',\bm{u}'} = \alpha_{\bm{u}}\quad
    \bm{u}\in\cube{r}
\]
For $ X\in\cube{r\times n}$ we define $K^S_{\bm{\alpha}}$ as follows:
\[
    K^{S}_{\bm{\alpha}}(\cf_X)
    \coloneqq 2^{|S|n}\mathcal{F}_{S}(L_{\bm{\alpha}})(X)
\]
The next proposition says that $K^{S}_{\bm{\alpha}}$ is a sparse 
product of lower-order Krawtchouks.

\begin{proposition}
For every $\bm{\beta}\in\mathcal{I}_{r,n}$ there holds
    \[
        K^{S}_{\bm{\alpha}}(\bm{\beta})
        = 
        \begin{cases}
        \prod K_{\bm{\alpha}^{S}_{\bm{v}}}(\bm{\beta}^{S}_{\bm{v}}) &
        \text{if } \bm{\alpha}^{S}_{\bm{v}}\cdot\1 = \bm{\beta}^{S}_{\bm{v}}\cdot\1~\forall \bm{v} \\
        0 & \text{otherwise.}
        \end{cases}
    \]
    where the product is over all $\bm{v}\in\cube{r-|S|}$.
    Here $\bm{\alpha}^{S}_v$ is the row of $\bm{\alpha}^{S}$ at index $\bm{v}$, and
    $\bm{\alpha}^{S}_{\bm{v}}\cdot\1 = \sum_{\bm{u}'\in\cube{|S|}}\alpha^{S}_{\bm{v},\bm{u}'}$.
    \label{prop:partial_Krawtchouk}
\end{proposition}
\begin{proof}

    Let $X\in\cube{r\times n}$ such that $\cf_X=\bm{\beta}$.
    By definition,
    \begin{dmath*}
        K^{S}_{\bm{\alpha}}(\bm{\beta})
        = \mathcal{F}_S(L_{\bm{\alpha}})(X)
        = 2^{-|S|n} \sum_{\y\in\cube{r\times n}}
            (-1)^{\langle X',\y' \rangle} \delta_{X''}(\y'')
            L_{\bm{\alpha}}(\y)
    \end{dmath*}
    We express $\delta_{X''}(\y'')L_{\bm{\alpha}}(\y)$ in terms of
    $\y'$, $\bm{\alpha}^{S}$ and $\bm{\beta}^{S}$.
    
    Let $\y\in\cube{r\times n}$ such that $\cf_\y=\bm{\alpha}$ and $\y'' = X''$.
    The number of times $\bm{u}''\in\cube{r-|S|}$ occurs in $\y''$ is $\bm{\alpha}^{S}_{\bm{u}''}\cdot\1= \sum_{\bm{u}'}\alpha^{S}_{\bm{u}'',\bm{u}'}$. But this is equal $\bm{\beta}^{S}_{\bm{u}''}\cdot\1$ because $\y''=X''$.
    
    Let $\y'|_{\y''=\bm{u}''}$ be the subset of columns from $\y'$
    for which the corresponding column in $\y''$ is $\bm{u}''$.
    Then $\cf_{\y'|_{\y''=\bm{u}''}} = \bm{\alpha}^{S}_{\bm{u}''}$.
    \begin{dmath*}
        L_{\bm{\alpha}}(\y)\delta_{X''}(\y'')=
        \prod_{\bm{u}''\in\cube{r-|S|}}
            \1_{[\bm{\alpha}^{S}_{\bm{u}''}\cdot\1=\bm{\beta}^{S}_{\bm{u}''}\cdot\1]}
            L_{\bm{\alpha}^{S}_{\bm{u}''}}(\y'|_{\y''=\bm{u}''})
    \end{dmath*}
    The sum over $\y$ can be broken into nested sums over
    $\{\y'|_{\y''=\bm{u}''}\}_{\bm{u}''\in\cube{r-|S|}}$,
    \[
        \sum_{\y\in\cube{r\times n}}
        = \sum_{\y'_0} \sum_{\y'_1} \dots \sum_{\y'_{2^{r-|S|}-1}}
    \]
    where $\y'_{\bm{u}''}\in\cube{|S|\times n}$ are mutually independent.
    Each factor in the product depends on a single $\y'_{\bm{u}''}$ so the order of
    summations and products can be reversed,
    \begin{multline*}
        \mathcal{F}_S(L_{\bm{\alpha}})(X)
        = 2^{-|S|n} \prod_{\bm{u}''\in\cube{r-|S|}} \bigg[
            \1_{[\bm{\alpha}^{S}_{\bm{u}''}\cdot\1=\bm{\beta}^{S}_{\bm{u}''}\cdot\1]}
            \times 
            \\
            \times
            \sum_{\y'_{\bm{u}''}\in\cube{|S|\times n}}
            (-1)^{\langle X'|_{X''=\bm{u}''},\y'_{\bm{u}''} \rangle} 
            L_{\bm{\alpha}^{S}_{\bm{u}''}}(\y'|_{\y''=\bm{u}''}) \bigg]
    \end{multline*}
    Observe that the inner sum is simply
    \[
    \hat{L}_{\bm{\alpha}^{S}_{\bm{u}''}}(X'|_{X''=\bm{u}''}) 
    = K^{S}_{\bm{\alpha}^{S}_{\bm{u}''}}(\bm{\beta}^{S}_{\bm{u}''}).
    \]
\end{proof}

In the last part of this section we consider
$K_{\bm{\alpha}}^{S}$ under the action of the general linear group $\GL{r}{2}$.

Fix $S\subset [r]$ and $\bm{\alpha}\in\mathcal{I}_{r,n}$.
It is easy to verify that 
\[
    K^{S}_{T\cdot \bm{\alpha}}(\cf_{X}) = 
        \mathcal{F}_{S}(L_{\bm{\alpha}}\circ T)(X)
\]
for every $X\in\cube{r\times n}$ and $T\in \GL{r}{2}$.
Propositions \ref{prop:Krawtchouk_under_Sr} and
\ref{prop:Krawtchouk_under_GLr} below are immediate consequences
of \ref{prop:partial_fourer_under_Sr} and
\ref{prop:partial_fourer_under_GLr}.

\begin{proposition}
    Let $T\in \GL{r}{2}$ be a permutation matrix, $T\bm{e}_i=\bm{e}_{\pi(i)}$ for
    some $\pi\in \mathfrak{S}_r$. Then
    \[
        K^{S}_{T \cdot \bm{\alpha}}
        = K^{\pi^{-1}(S)}_{\bm{\alpha}} \circ T
    \]
    \label{prop:Krawtchouk_under_Sr}
\end{proposition}

\begin{proposition}
    Let $T\in \GL{r}{2}$ be the mapping $\bm{e}_i\mapsto \bm{e}_i+\bm{e}_j$
    for some $i,j\in[r]$, and
    $\bm{e}_k\mapsto \bm{e}_k$ for every $k\neq i$.
    \begin{itemize}
        \item if $i,j\in S$:
        \[
            K^{S}_{T \cdot \bm{\alpha}}
            = K^{S}_{\bm{\alpha}} \circ T^{\intercal}
        \]
        \item if $i,j\notin S$:
        \[
            K^{S}_{T \cdot \bm{\alpha}}
            = K^{S}_{\bm{\alpha}} \circ T
        \]
        \item if $i\in S,j\notin S$:
        \[
            K^{S}_{T \cdot \bm{\alpha}}(\bm{\beta})
            = (-1)^{\sum_{\bm{u}\in\cube{r}} \beta_{\bm{u}} u_i u_j}
                K^{S}_{\bm{\alpha}}(\bm{\beta})
        \]
    \end{itemize}
    \label{prop:Krawtchouk_under_GLr}
\end{proposition}

\bibliography{refs}
\bibliographystyle{IEEEtran}

\appendix

\section{Appendix}
\label{section:appendix}

\subsection{Proofs for section
\ref{section:partial_fourier_transform_symmetries}}

\begin{proof}[Proof of Proposition \ref{prop:partial_fourer_under_Sn}]
    For $\bm{x},\bm{y}\in\cube{n}$,
    \begin{dmath*}
        \chi_{\bm{x}}(\sigma\cdot {\bm{y}})
        = (-1)^{\sum_{i=1}^{n} x_i y_{\sigma(i)}}
        = (-1)^{\sum_{i=1}^{n} x_{\sigma^{-1}(i)} y_{i}}
        = \chi_{\sigma^{-1}\cdot {\bm{x}}}({\bm{y}})
    \end{dmath*}
    and
    \[
        \delta_{\bm{x}}(\sigma\cdot {\bm{y}})
        = \prod_{i=1}^{n}\delta_{x_i}(y_{\sigma(i)})
        = \prod_{i=1}^{n}\delta_{x_{\sigma^{-1}(i)}}(y_{i})
        = \delta_{\sigma^{-1}\cdot {\bm{x}}}({\bm{y}})
    \]
    Hence, for $\x,\y\in\cube{r\times n}$,
    \[
        \chi^{S}_{\x}(\sigma \cdot \y)
        =  \chi^{S}_{\sigma^{-1}\cdot \x}(\y)
    \]
    Finally, letting $\y'=\sigma\cdot \y$,
    \begin{dmath*}
        \mathcal{F}_S(f\circ \sigma)(\x)
        = \sum_{\y\in\cube{r\times n}}\chi_{\x}(\y)f(\sigma\cdot \y)
        = \sum_{\y'\in\cube{r\times n}}\chi_{\x}(\sigma^{-1}\cdot\y')f(\y')
        \\
        = \sum_{\y'\in\cube{r\times n}}\chi_{\sigma\cdot\x}(\y')f(\y')
        = \mathcal{F}_S(f)(\sigma\cdot\x)
    \end{dmath*}
\end{proof}

\begin{proof}[Proof of Proposition \ref{prop:partial_fourer_under_Sr}]
    \begin{dmath*}
        \mathcal{F}_S(f\circ \pi)({\bm{x}}_1,\dots,{\bm{x}}_r)
        = 2^{-rn} \sum_{{\bm{y}}_1,\dots,{\bm{y}}_r}
            \chi^{S}_{{\bm{x}}_1,\dots,{\bm{x}}_r}({\bm{y}}_1,\dots,{\bm{y}}_r)
            f({\bm{y}}_{\pi(1)},\dots,{\bm{y}}_{\pi(r)})
        = 2^{-rn} \sum_{{\bm{y}}_1,\dots,{\bm{y}}_r}
            \chi^{S}_{{\bm{x}}_1,\dots,{\bm{x}}_r}({\bm{y}}_{\pi^{-1}(1)},\dots,{\bm{y}}_{\pi^{-1}(r)})
            \times
            \\
            \times
            f({\bm{y}}_1,\dots,{\bm{y}}_r)
        = 2^{-rn} \sum_{{\bm{y}}_1,\dots,{\bm{y}}_r}
            \prod_{i\in S}\chi_{{\bm{x}}_i}({\bm{y}}_{\pi^{-1}(i)})
            \times
            \\
            \times
            \prod_{i\in [r]\setminus S}\delta_{{\bm{x}}_i}({\bm{y}}_{\pi^{-1}(i)})
            f({\bm{y}}_1,\dots,{\bm{y}}_r)
        = 2^{-rn} \sum_{{\bm{y}}_1,\dots,{\bm{y}}_r}
            \prod_{j\in \pi^{-1}(S)}\chi_{{\bm{x}}_{\pi(j)}}({\bm{y}}_{j})
            \times
            \\
            \times
            \prod_{j\in [r]\setminus \pi^{-1}(S)}\delta_{{\bm{x}}_{\pi(j)}}({\bm{y}}_{j})
            f({\bm{y}}_1,\dots,{\bm{y}}_r)
        = \mathcal{F}_{\pi^{-1}(S)}(f)({\bm{x}}_{\pi(1)},\dots,{\bm{x}}_{\pi(r)})
    \end{dmath*}
\end{proof}

\begin{proof}[Proof of Proposition \ref{prop:partial_fourer_under_GLr}]
    Consider $f$ as a function of $r$ vectors,
    ${\bm{x}}_1,\dots,{\bm{x}}_r\in\cube{n}$. Then $T$ maps ${\bm{x}}_i\mapsto
    {\bm{x}}_i+{\bm{x}}_j$, and ${\bm{x}}_k\mapsto {\bm{x}}_k$ for $k\neq i$.

    For $k=1,\dots,r$, and ${\bm{x}}\in\cube{n}$, define a set of functions
    $\{\psi^{(k)}_{\bm{x}}\}_{k\in[r]}$, where
    $\psi^{(k)}_{\bm{x}}=\chi_{\bm{x}}$ if $k\in S$ and
    $\psi^{(k)}_{\bm{x}}=\delta_{\bm{x}}$ otherwise.
    \begin{dmath*}
        \mathcal{F}_S(f\circ T)({\bm{x}}_1,\dots,{\bm{x}}_r) =
        2^{-rn} \sum_{\bm{y}_1,\dots,\bm{y}_r} \prod_{k=1}^{r}\psi^{(k)}_{\bm{y}_k}(\bm{y}_k)
        f(\bm{y}_1,\dots,\underbrace{\bm{y}_i+\bm{y}_j}_{\text{index } i},\dots,\bm{y}_r)
    \end{dmath*}
    replacing the sum over $\bm{y}_i$ by a sum over $\bm{y}_i+\bm{y}_j$, we get
    \begin{dmath}
        =
        2^{-rn} \sum_{\bm{y}_1,\dots,\bm{y}_r}
            \psi^{(i)}_{\bm{x}_i}(\bm{y}_i+\bm{y}_j)\psi^{(j)}_{\bm{y}_j}(\bm{y}_j)
            \prod_{k\in[r]\setminus\{i,j\}}\psi^{(k)}_{\bm{y}_k}(\bm{y}_k) 
                f(\bm{y}_1,\dots,\bm{y}_i,\dots,\bm{y}_r)
        \label{eq:partial_fourier_GL}
    \end{dmath}
    We now examine the expression
    $\psi^{(i)}_{\bm{x}_i}(\bm{y}_i+\bm{y}_j)\psi^{(j)}_{\bm{y}_j}(\bm{y}_j)$
    in the different cases of the proposition.
    \begin{itemize}
        \item $i,j\in S$:
        \begin{dmath*}
            \psi^{(i)}_{\bm{x}_i}(\bm{y}_i+\bm{y}_j)\psi^{(j)}_{\bm{x}_j}(\bm{y}_j)
            = \chi_{\bm{x}_i}(\bm{y}_i+\bm{y}_j)\chi_{\bm{x}_j}(\bm{y}_j)
            = \chi_{\bm{x}_i}(\bm{y}_i)\chi_{\bm{x}_i+\bm{x}_j}(\bm{y}_j)
            = \psi^{(i)}_{\bm{x}_i}(\bm{y}_i)\psi^{(j)}_{\bm{x}_i+\bm{x}_j}(\bm{y}_j)
        \end{dmath*}
        which is the same as applying the mapping $\bm{x}_j\mapsto \bm{x}_i+\bm{x}_j$, or
        equivalently $T^\intercal$, to ${\bm{x}}_1,\dots,{\bm{x}}_r$.
        \item $i,j\notin S$:
        \begin{dmath*}
            \psi^{(i)}_{\bm{x}_i}(\bm{y}_i+\bm{y}_j)\psi^{(j)}_{\bm{x}_j}(\bm{y}_j)
            = \delta_{\bm{x}_i}(\bm{y}_i+\bm{y}_j)\delta_{\bm{x}_j}(\bm{y}_j)
            = \delta_{\bm{x}_i}(\bm{y}_i+\bm{x}_j)\delta_{\bm{x}_j}(\bm{y}_j)
            = \delta_{\bm{x}_i+\bm{x}_j}(\bm{y}_i)\delta_{\bm{x}_j}(\bm{y}_j)
            = \psi^{(i)}_{\bm{x}_i+\bm{x}_j}(\bm{y}_i)\psi^{(j)}_{\bm{x}_j}(\bm{y}_j)
        \end{dmath*}
        which equivalent to applying $T$ to ${\bm{x}}_1,\dots,{\bm{x}}_r$.
        \item $i\in S, j\notin S$:
        \begin{dmath*}
            \psi^{(i)}_{\bm{x}_i}(\bm{y}_i+\bm{y}_j)\psi^{(j)}_{\bm{x}_j}(\bm{y}_j)
            = \chi_{\bm{x}_i}(\bm{y}_i+\bm{y}_j)\delta_{\bm{x}_j}(\bm{y}_j)
            = \chi_{\bm{x}_i}(\bm{x}_j)\chi_{\bm{x}_i}(\bm{y}_i)\delta_{\bm{x}_j}(\bm{y}_j)
            = \chi_{\bm{x}_i}(\bm{x}_j)\psi^{(i)}_{\bm{x}_i}(\bm{y}_i)\psi^{(j)}_{\bm{x}_j}(\bm{y}_j)
        \end{dmath*}
        observe that $\chi_{\bm{x}_i}(\bm{x}_j)$ is constant with respect to the sum
        in \eqref{eq:partial_fourier_GL}.
        \item $i\notin S, j\in S$:
        \[
            \psi^{(i)}_{\bm{x}_i}(\bm{y}_i+\bm{y}_j)\psi^{(j)}_{\bm{x}_j}(\bm{y}_j)
            = \delta_{\bm{x}_i}(\bm{y}_i+\bm{y}_j)\chi_{\bm{x}_j}(\bm{y}_j)
        \]
        Here there is no obvious way to rewrite the functions so as to separate
        $\bm{y}_i$ and $\bm{y}_j$.
    \end{itemize}
\end{proof}

\subsection{Numerical Results}\label{section:numerical_results}

We have experimented with several variants of $\DelsarteLin{2}{n}{d}$ with $n$
ranging between $10$ and $40$ and $d\le n/2$ is even.
In those variants of the LP - we replace 
\eqref{eq:delsarte_lin:C2} with
\eqref{eq:delsarte_lin:C2_weak}, and 
\eqref{eq:delsarte_lin:objective} with 
\eqref{eq:delsarte_lin:objective_r}.
The table below only shows the
results for $n\geq 20$.

The number of variables is
$|\mathcal{I}_{r,n}|=\binom{n+2^r-1}{2^r-1}$
and, if we
consider $r$ as a constant, there are $O_n(|\mathcal{I}_{r,n}|)$ 
constraints. In practice, symmetrization w.r.t.
$\GL{r}{2}$ reduces the problem size
(variables $\times$ constraints)
by a factor of $2^{\Omega(r)}$, which is
significant. Since
we have not yet developed the necessary theoretical tools for
such symmetrization, it was carried out algorithmically. We intend to
develop such theory so as to solve instances of 
$\DelsarteLin{r}{n}{d}$ with larger values of $r$.

The number of variables is further reduced using the well-known fact,
that if $d$ is even then an even code attains $A(n,d)$.
Namely, we set $\varphi_{\bm{\alpha}} = 0$
if $(n- \chi_{\bm{u}}^\intercal \bm{\alpha})/2$ is odd, for some $\bm{u}\in\cube{r}$.

The Krawtchouk polynomials were computed with a
recurrence formula, e.g. (16) in
\cite{diaconis2014introduction}. The partial Krawtchouks $K^{S}_{\bm{\alpha}}$
were computed using proposition \ref{prop:partial_Krawtchouk}.
We used two exact solvers: SoPlex
\cite{gleixner2012improving,gleixner2016iterative,gamrath2020scip}
and QSoptEx \cite{applegate2007exact}, with up to 128GB of RAM
and at most 3 days of runtime. Some instances were solved 
by one solver and not the other.
Missing entries were solved by neither.

The best in each row is marked with boldface.
Entries are marked with a "$*$" if
$\lfloor \log_2(\text{entry}) \rfloor$ equals the best
known upper bound, as reported in
\cite{Grassl:codetables}.

\begin{center}

\setlength{\tabcolsep}{3pt}
\def\arraystretch{1.3}

\tablefirsthead{
\toprule
   \multicolumn{2}{r}{Variant} & \multicolumn{2}{l}{$(C2')$} & \multicolumn{2}{l}{$(C2)$} &             Delsarte & Schrijver \cite{schrijver2005new} \\
   &  &              $(Obj')$ &               $(Obj)$ &              $(Obj')$ &               $(Obj)$ &         &      \\
n & dist. &                       &                       &                       &                       &         &             \\
\midrule}

\tablehead{
\toprule
   \multicolumn{2}{r}{Variant} & \multicolumn{2}{l}{$(C2')$} & \multicolumn{2}{l}{$(C2)$} &             Delsarte & Schrijver \cite{schrijver2005new} \\
   &  &              $(Obj')$ &               $(Obj)$ &              $(Obj')$ &               $(Obj)$ &        &       \\
n & dist. &                       &                       &                       &                       &          &            \\
\midrule}

\tabletail{
\midrule
\multicolumn{8}{r}{{Continued on next page}} \\
\midrule}

\tablelasttail{\bottomrule}

\begin{supertabular}{llllllll}
20 & 4  &             $26214^*$ &             $26214^*$ &    $\textbf{21845}^*$ &    $\textbf{21845}^*$ &            $26214^*$ & -\\
   & 6  &                $2328$ &                $2285$ &     $\textbf{1588}^*$ &              $1593^*$ &               $2373$  & - \\
   & 8  &               $268^*$ &      $\textbf{256}^*$ &      $\textbf{256}^*$ &      $\textbf{256}^*$ &              $291^*$  & $274^*$ \\
   & 10 &                  $40$ &                  $40$ &       $\textbf{24}^*$ &       $\textbf{24}^*$ &                 $40$  & - \\
\midrule
21 & 4  &    $\textbf{43691}^*$ &    $\textbf{43691}^*$ &    $\textbf{43691}^*$ &    $\textbf{43691}^*$ &            $47663^*$  & - \\
   & 6  &                $4197$ &                $4138$ &              $3010^*$ &     $\textbf{2977}^*$ &               $4443$  & - \\
   & 8  &      $\textbf{512}^*$ &      $\textbf{512}^*$ &      $\textbf{512}^*$ &      $\textbf{512}^*$ &              $572^*$  & - \\
   & 10 &                $51^*$ &                $52^*$ &       $\textbf{35}^*$ &                $36^*$ &                 $64$  & - \\
\midrule
22 & 4  &    $\textbf{87381}^*$ &    $\textbf{87381}^*$ &    $\textbf{87381}^*$ &    $\textbf{87381}^*$ &   $\textbf{87381}^*$  & - \\
   & 6  &              $7380^*$ &              $7327^*$ &              $5770^*$ &     $\textbf{5608}^*$ &             $7724^*$  & - \\
   & 8  &     $\textbf{1024}^*$ &     $\textbf{1024}^*$ &     $\textbf{1024}^*$ &     $\textbf{1024}^*$ &    $\textbf{1024}^*$  & - \\
   & 10 &                  $92$ &                  $89$ &       $\textbf{57}^*$ &                $61^*$ &                 $95$  & $87$ \\
\midrule
23 & 4  &   $\textbf{174763}^*$ &   $\textbf{174763}^*$ &   $\textbf{174763}^*$ &   $\textbf{174763}^*$ &  $\textbf{174763}^*$  & - \\
   & 6  &             $13703^*$ &             $13690^*$ &             $10447^*$ &    $\textbf{10102}^*$ &            $13776^*$  & $13766^*$ \\
   & 8  &     $\textbf{2048}^*$ &     $\textbf{2048}^*$ &     $\textbf{2048}^*$ &     $\textbf{2048}^*$ &    $\textbf{2048}^*$  & - \\
   & 10 &                 $152$ &                 $152$ &       $\textbf{90}^*$ &                $93^*$ &                $152$  & - \\
\midrule
24 & 4  &   $\textbf{349525}^*$ &   $\textbf{349525}^*$ &   $\textbf{349525}^*$ &   $\textbf{349525}^*$ &  $\textbf{349525}^*$  & - \\
   & 6  &             $24054^*$ &             $24018^*$ &             $18786^*$ &    $\textbf{18715}^*$ &            $24108^*$  & - \\
   & 8  &     $\textbf{4096}^*$ &     $\textbf{4096}^*$ &     $\textbf{4096}^*$ &     $\textbf{4096}^*$ &    $\textbf{4096}^*$  & - \\
   & 10 &                 $280$ &                 $280$ &      $\textbf{155}^*$ &               $160^*$ &                $280$  & - \\
   & 12 &       $\textbf{48}^*$ &       $\textbf{48}^*$ &       $\textbf{48}^*$ &       $\textbf{48}^*$ &      $\textbf{48}^*$  & - \\
\midrule
25 & 4  &            $599186^*$ &            $599186^*$ &            $582826^*$ &   $\textbf{579701}^*$ &           $645278^*$  & - \\
   & 6  &               $47481$ &               $47176$ &      $\textbf{34657}$ &               $34729$ &              $48149$  & $47998$ \\
   & 8  &              $5666^*$ &              $5571^*$ &              $4450^*$ &     $\textbf{4200}^*$ &             $6475^*$  & $5477^*$ \\
   & 10 &                 $511$ &                 $497$ &        $\textbf{262}$ &                 $284$ &                $551$  & $503$ \\
   & 12 &                $61^*$ &                $62^*$ &                $49^*$ &       $\textbf{48}^*$ &                 $75$  & - \\
\midrule
26 & 4  &           $1198373^*$ &           $1198373^*$ &           $1126532^*$ &  $\textbf{1121065}^*$ &          $1198373^*$  & - \\
   & 6  &               $86693$ &               $86847$ &      $\textbf{66014}$ &               $66638$ &              $93623$  & - \\
   & 8  &               $10099$ &               $10031$ &              $7516^*$ &     $\textbf{7508}^*$ &              $10435$  & - \\
   & 10 &                 $930$ &                 $922$ &      $\textbf{490}^*$ &                 $533$ &               $1040$  & $886$ \\
   & 12 &               $105^*$ &                $99^*$ &       $\textbf{77}^*$ &                $79^*$ &              $113^*$  & - \\
\midrule
27 & 4  &           $2396745^*$ &           $2396745^*$ &  $\textbf{2097152}^*$ &  $\textbf{2097152}^*$ &          $2396745^*$  & - \\
   & 6  &              $162180$ &              $162027$ &   $\textbf{125238}^*$ &            $126201^*$ &             $163840$  & - \\
   & 8  &               $17803$ &               $17727$ &    $\textbf{12698}^*$ &             $12774^*$ &              $18190$  & $17768$ \\
   & 10 &                $1766$ &                $1766$ &      $\textbf{854}^*$ &               $954^*$ &               $1766$  & - \\
   & 12 &               $171^*$ &               $171^*$ &               $132^*$ &      $\textbf{129}^*$ &              $171^*$  & - \\
\midrule
28 & 4  &           $4793490^*$ &           $4793490^*$ &  $\textbf{4194304}^*$ &                     - &          $4793490^*$  & - \\
   & 6  &              $291202$ &              $291173$ &            $234649^*$ &   $\textbf{234626}^*$ &             $291271$  & - \\
   & 8  &             $32126^*$ &             $32119^*$ &    $\textbf{21989}^*$ &             $22120^*$ &            $32206^*$  & $32151^*$ \\
   & 10 &                $3194$ &                $3189$ &     $\textbf{1482}^*$ &              $1646^*$ &               $3200$  & - \\
   & 12 &                 $288$ &                 $288$ &      $\textbf{213}^*$ &      $\textbf{213}^*$ &                $288$  & - \\
   & 14 &                $56^*$ &                $56^*$ &       $\textbf{32}^*$ &       $\textbf{32}^*$ &               $56^*$  & - \\
\midrule
29 & 4  & $\textbf{8388608}^*$ &  $\textbf{8388608}^*$ &  $\textbf{8388608}^*$ &  $\textbf{8388608}^*$ &          $8947849^*$  & - \\
   & 6  &              $574493$ &              $573756$ &   $\textbf{430773}^*$ &            $432499^*$ &             $581827$  & - \\
   & 8  &               $57247$ &               $57217$ &      $\textbf{38276}$ &               $39205$ &              $58097$  & - \\
   & 10 &                $6155$ &                $6074$ &       $\textbf{2743}$ &                $3181$ &               $6363$  & - \\
   & 12 &                 $550$ &                 $541$ &      $\textbf{320}^*$ &               $323^*$ &                $573$  & - \\
   & 14 &                  $70$ &                  $72$ &       $\textbf{47}^*$ &                $49^*$ &                 $88$  & - \\
\midrule
30 & 8  &              $108267$ &              $107044$ &      $\textbf{67353}$ &               $71095$ &             $114816$  & - \\
   & 10 &               $11517$ &               $11340$ &       $\textbf{4827}$ &                $5929$ &              $12525$  & - \\
   & 12 &              $1022^*$ &                $1026$ &      $\textbf{535}^*$ &               $582^*$ &               $1132$  & - \\
   & 14 &               $114^*$ &               $112^*$ &       $\textbf{74}^*$ &                $80^*$ &                $130$  & - \\
\midrule
31 & 10 &               $20838$ &               $20738$ &       $\textbf{8651}$ &                     - &              $22296$  & - \\
   & 12 &              $1781^*$ &              $1763^*$ &     $\textbf{1024}^*$ &     $\textbf{1024}^*$ &             $1840^*$  & - \\
   & 14 &                 $196$ &                 $196$ &      $\textbf{110}^*$ &               $115^*$ &                $196$  & - \\
\midrule
32 & 12 &              $3082^*$ &              $3082^*$ &     $\textbf{2048}^*$ &     $\textbf{2048}^*$ &             $3082^*$  & - \\
   & 14 &                 $314$ &                 $313$ &      $\textbf{187}^*$ &               $191^*$ &                $315$  & - \\
   & 16 &       $\textbf{64}^*$ &       $\textbf{64}^*$ &       $\textbf{64}^*$ &       $\textbf{64}^*$ &      $\textbf{64}^*$  & - \\
\midrule
33 & 12 &                $5821$ &                $5821$ &     $\textbf{2903}^*$ &              $3037^*$ &               $5821$  & - \\
   & 14 &                 $617$ &                 $612$ &      $\textbf{310}^*$ &               $324^*$ &                $629$  & - \\
   & 16 &                $80^*$ &                $82^*$ &                $65^*$ &       $\textbf{64}^*$ &               $99^*$  & - \\
\midrule
34 & 12 &               $10878$ &               $10671$ &                     - &     $\textbf{5726}^*$ &              $11641$  & - \\
   & 14 &                $1195$ &                $1203$ &      $\textbf{510}^*$ &                 $568$ &               $1258$  & - \\
   & 16 &               $124^*$ &               $120^*$ &       $\textbf{99}^*$ &               $103^*$ &                $144$  & - \\
\midrule
35 & 12 &               $20011$ &               $19810$ &                     - &      $\textbf{10603}$ &              $21727$  & - \\
   & 14 &                $1774$ &                $1759$ &      $\textbf{890}^*$ &               $989^*$ &               $2026$  & - \\
   & 16 &               $215^*$ &               $215^*$ &               $172^*$ &      $\textbf{169}^*$ &              $215^*$  & - \\
\midrule
36 & 14 &                     - &                     - &     $\textbf{1664}^*$ &              $1786^*$ &               $3177$  & - \\
   & 16 &               $352^*$ &               $352^*$ &      $\textbf{256}^*$ &      $\textbf{256}^*$ &              $352^*$  & - \\
   & 18 &                  $72$ &                  $72$ &         $\textbf{40}$ &         $\textbf{40}$ &                 $72$  & - \\
\midrule
37 & 16 &                 $704$ &                 $704$ &      $\textbf{366}^*$ &                     - &                $704$  & - \\
\midrule
38 & 14 &                     - &               $10211$ &                     - &     $\textbf{5348}^*$ &              $10211$  & - \\
\midrule
39 & 16 &                $1905$ &                $1918$ &       $\textbf{1138}$ &       $\textbf{1138}$ &               $2271$  & - \\
\midrule
40 & 16 &                $3493$ &                $3488$ &                     - &       $\textbf{2276}$ &               $3510$  & - \\
\end{supertabular}

\begin{figure}[t]
\centering
\makeatletter
\@for\dval:={6,8,10,12,14,16}\do{
\begin{subfigure}[b]{0.3\textwidth}
\centering
\begin{adjustbox}{width=\linewidth}
\begin{tikzpicture}
\begin{axis}[
    legend style={at={(0,1.3)},anchor=north west,legend cell align=left,font=\large},
    grid=major,
]
    \addplot[ultra thick,color=color1, mark=o]
        table [x=n, y=delsarte, col sep=comma] {data/diff_plot_table_d_\dval_floor.csv};
    \ifnum \dval=6 {\addlegendentry{Delsarte}} \else {} \fi
    \addplot[ultra thick,color=color3, mark=square]
        table [x=n, y=strong_lp, col sep=comma] {data/diff_plot_table_d_\dval_floor.csv};
    \ifnum \dval=6 {\addlegendentry{$\DelsarteLinCube{2}{n}{d}$}} \else {} \fi
    \addplot[ultra thick,color=color2, mark=x]
        table [x=n, y=weak_lp, col sep=comma] {data/diff_plot_table_d_\dval_floor.csv};
    \ifnum \dval=6 {\addlegendentry{ \textup{KrawtchoukLin}$(n,d,2)$ \cite{coregliano2021complete}}} \else {} \fi
\end{axis}
\end{tikzpicture}
\end{adjustbox}
\caption{$d=\dval$}
\end{subfigure}
}
\makeatother
\caption{
(Lower is better).
Similar to Figure \ref{fig:results_diff_plot},
stated in terms of the code's {\em dimension}.
Each point represents $\lfloor \log_2(LP(n,d)) \rfloor - \textup{bestKnown}(n,d)$,
where LP is one of: $\Delsarte{n}{d}$, $\DelsarteLinCube{2}{n}{d}$, and 
\textup{KrawtchoukLin}$(n,d,2)$, and bestKnown$(n,d)$ is the 
best known upper bound on $\log_2(A^\Lin(n,d))$.
}
\label{fig:results_diff_plot_floor}
\end{figure}
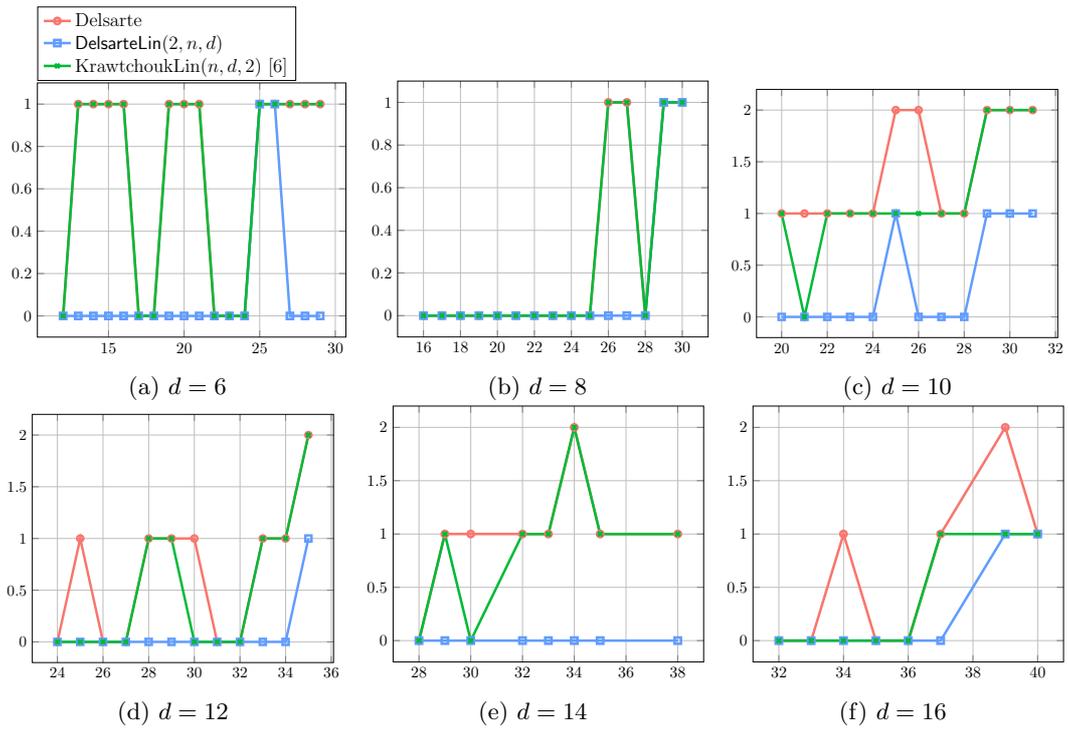
\end{center}

\end{document}